  \documentclass[11pt, draftclsnofoot, journal, onecolumn]{IEEEtran}
	\usepackage{setspace}
  \usepackage[dvips]{graphicx}
  \usepackage{amsmath, amsthm, amssymb, amsfonts, mathtools}
  \usepackage{cases, multirow}
  \usepackage{srcltx,color}
  \usepackage{cite, url}
  \usepackage[colorinlistoftodos,prependcaption,textsize=tiny]{todonotes}
  \usepackage{algorithm, algpseudocode, cool}
  \usepackage{xcolor}

\usepackage{subfigure}

  \long\def\comment#1{}

  \newtheorem{thm}{Theorem}
  \newtheorem{lemm}{Lemma}
  
  \newtheorem{assumption}{Assumption}

  \newtheorem{remark}{Remark}

  \def\figref#1{Fig.~\ref{#1}}
  \def\be{\begin{equation} }
  \def\ee{\end{equation} }

  \newcommand\mycom[2]{\genfrac{}{}{0pt}{}{#1}{#2}}
  \newlength{\myimagewidth}

  \title{A Stochastic Geometric Analysis of Device-to-Device Communications Operating over Generalized Fading Channels}

\begin{document}
  \setlength{\myimagewidth}{\dimexpr\linewidth/2-1em\relax}

  \author{
     \IEEEauthorblockN{Young Jin Chun, Simon L. Cotton, Harpreet S. Dhillon, Ali Ghrayeb, and Mazen O. Hasna}
     \thanks{Y. J. Chun and S. L. Cotton are with the Wireless Communications Laboratory, ECIT Institute, Queens University Belfast, United Kingdom.
     H. S. Dhillon is with Wireless@VT, Department of Electrical and Computer Engineering, Virginia Tech, Blacksburg, VA, USA.
     A. Ghrayeb is with Department of Electrical and Computer Engineering, Texas A\&M University at Qatar, Doha, Qatar. M. O. Hasna is with Department of Electrical Engineering, Qatar University, Doha, Qatar.
     (Email: Y.Chun@qub.ac.uk, simon.cotton@qub.ac.uk, hdhillon@vt.edu, ali.ghrayeb@qatar.tamu.edu, hasna@qu.edu.qa).}
     \thanks{This work has been submitted to the IEEE for possible publication. Copyright may be transferred without notice, after which this version may no longer be accessible.}
   }

  \maketitle

  \begin{abstract}
	\begin{spacing}{1.2}
  Device-to-device (D2D) communications are now considered as an integral part of future 5G networks which will enable direct communication between user equipment (UE) without unnecessary routing via the network infrastructure. This architecture will result in higher throughputs than conventional cellular networks, but with the increased potential for co-channel interference induced by randomly located cellular and D2D UEs. The physical channels which constitute D2D communications can be expected to be complex in nature, experiencing both line-of-sight (LOS) and non-LOS (NLOS) conditions across closely located D2D pairs. As well as this, given the diverse range of operating environments, they may also be subject to clustering of the scattered multipath contribution, \textit{i.e.}, propagation characteristics which are quite dissimilar to conventional Rayeligh fading environments. To address these challenges, we consider two recently proposed generalized fading models, namely $\kappa-\mu$ and $\eta-\mu$, to characterize the fading behavior in D2D communications. Together, these models encompass many of the most widely encountered and utilized fading models in the literature such as Rayleigh, Rice (Nakagami-$n$), Nakagami-$m$, Hoyt (Nakagami-$q$) and One-Sided Gaussian. Using stochastic geometry we evaluate the rate and bit error probability of D2D networks under generalized fading conditions. Based on the analytical results, we present new insights into the trade-offs between the reliability, rate, and mode selection under realistic operating conditions. Our results suggest that D2D mode achieves higher rates over cellular link at the expense of a higher bit error probability. Through numerical evaluations, we also investigate the performance gains of D2D networks and demonstrate their superiority over traditional cellular networks.
  \end{spacing} 
\end{abstract}


  \begin{IEEEkeywords}
     Device-to-device network, $\eta-\mu$ fading, $\kappa-\mu$ fading, rate-reliability trade-off, stochastic geometry.
  \end{IEEEkeywords}

  \section{Introduction}

  \subsection{Related Works}

  The recent unprecedented growth in mobile traffic has compelled the telecommunications industry to come up with new and innovative ways to improve cellular network performance to meet the ever increasing data demands. This has led to the introduction of the fifth generation (5G) of networks which are expected to provide $1000$ fold gains in capacity while achieving latencies of less than $1$ millisecond \cite{Nokia2014}. Device-to-device communications are a strong contender for 5G networks \cite{Tehrani2014} that allow direct communication between user equipments (UEs) without unnecessary routing of traffic through the network infrastructure, resulting in shorter transmission distances and improved data rates than the traditional cellular networks \cite{Asadi2014}.

  Currently, D2D is standardized by the 3rd Generation Partnership Project (3GPP) in LTE Release $12$ to provide proximity based services and public safety applications \cite{3Gpp2013}. In parallel to the standardization efforts, D2D communications have been actively studied by the research community. For example, in \cite{Lin2000}, the authors have proposed D2D as a multi-hop scheme, while in \cite{Kaufman2008, Doppler2009}, the work conducted in \cite{Lin2000} has been extended to prove that D2D communications can improve spectral efficiency and the coverage of conventional cellular networks. Additionally, D2D has also been applied to multi-cast scenarios \cite{Zhou2013}, machine-to-machine (M2M) communications \cite{Lien2011}, and cellular off-loading \cite{Bao2010}.

  While D2D communications offer many advantages, they also come with numerous challenges. These include the difficulties in accurately modeling the interference induced by cellular and D2D UEs, and consequently optimizing the resource allocation based on the interference model. Most of the previous works published in this area have relied on system-level simulations with a large parameter set \cite{Mustafa2014}, meaning that it is difficult to draw general conclusions. Recently, stochastic geometry has received considerable attention as a useful mathematical tool for interference modeling. Specifically, stochastic geometry treats the locations of the interferer as points distributed according to a spatial point process \cite{Haenggi2013}. Such an approach captures the topological randomness in the network geometry, offers high analytical flexibility and achieves accurate performance evaluation \cite{Andrews2010, Andrews2011, Jo2012, Dhillon2012, Chun2015}.

  Much work has also been done on evaluating the performance of D2D networks over Rayleigh fading channels. In \cite{Lin2014}, the authors have compared two D2D spectrum sharing schemes (overlay and underlay) and evaluated the average achievable rate for each scheme based on a stochastic geometric framework. In \cite{George2014}, the authors extended the work conducted in \cite{Lin2014} by considering a D2D link whose length depends on the user density. Flexible mode selections have also been considered, in \cite{Elsawy2014} a novel strategy is proposed which makes use of truncated channel inversion based power control for underlaid D2D networks. Notwithstanding these advances, limited work has been conducted to consider D2D networks with general fading channels, for example in \cite{Peng2014}, the authors have considered underlaid D2D networks over Rician fading channels and evaluated the success probability and average achievable rate.

  \subsection{Motivation and Contributions}

  In 5G networks and especially for D2D communications, fading environments will range from homogeneous and circularly symmetric through to non-homogeneous and non-circularly symmetric. For example, the METIS project has already demonstrated that the physical channel of 5G networks can be inhomogeneous with clusters of non-circularly symmetric scattered waves \cite{Medbo2014}. Clearly in this case, the assumption of traditional, homogeneous, linear and single cluster fading models such as Rayleigh  will no longer be sufficient and we must look towards other more general and realistic models such as $\kappa-\mu$ \cite{Yacoub2007a, DaCosta2008, 6953146} and $\eta-\mu$ \cite{Yacoub2007a, 4162456}. Influenced by this, we consider the $\kappa-\mu$ fading model which accounts for homogeneous, linear environments with line-of-sight (LOS) components and multiple clusters of scattered signal contributions, while the $\eta-\mu$ fading model represents inhomogeneous, linear environments with non-line-of-sight (NLOS) conditions and multiple clusters of scattered signal contributions.

  As discussed earlier, most of the existing published work in stochastic geometry for general wireless networks has been focused on Rayleigh fading environments, owing to its tractability and favorable analytical characteristics. The SINR distributions for general fading environments require evaluating the sum-products of aggregate interference where several approaches have been proposed to facilitate the derivation. One approach is the conversion method, which is utilized in \cite{Blaszczyszyn2013,Keeler2013, Madhusudhanan2014a, Dhillon2014, Zhang2014}, that treats the channel randomness as a perturbation in the location of the transmitter and then used the displacement theorem to transform the original network with general fading into an equivalent network without fading. The conversion method can incorporate an arbitrary fading distribution, but is not applicable when there is an exclusion zone in the interference field. An alternative approach is to express the interference functionals as an infinite series \cite{Peng2014,Tanbourgi2014}. The series representation method can be applied to an arbitrary network model, nonetheless, it may be difficult (or impossible) to analytically find the closed form expression for the high order derivatives of the Laplace transform.

  Motivated by these approaches and their limitations, we propose a stochastic geometric framework to facilitate the performance evaluation of D2D networks over generalized fading channels; \textit{i.e.}, $\kappa-\mu$ and $\eta-\mu$. We consider a D2D network overlaid upon a cellular network where the spatial locations of the mobile UEs as well as the base stations (BSs) are modeled as Poisson point process (PPP). The proposed framework can evaluate the average of an arbitrary function of the SINR over generalized fading channels, thereby enabling the estimation of the average rate, outage probability, and bit error probability.
  
  The main contributions of this paper may be summarized as follows.
  \begin{enumerate}
  	  \item We characterize the distribution of the interference experienced by cellular and D2D UEs for various generalized fading environments, namely, (i) $\kappa-\mu$ and (ii) $\eta-\mu$ fading. It is worth highlighting that these two models together encompass all of the most popular fading models proposed in the literature, including Rayleigh, Rice (Nakagami-$n$), Nakagami-$m$, Hoyt (Nakagami-$q$), and One-Sided Gaussian to name but a few.

  	  \item We also introduce a novel stochastic geometric approach for evaluating the performance of D2D networks over generalized fading channels. This approach enable us to evaluate the average of an arbitrary function of the SINR as a closed form expression.

  	  \item We invoke the proposed stochastic geometric approach to evaluate the average rate and bit error probability of D2D networks and compare that to the performance of conventional cellular networks. Furthermore, we also study the trade-off among a number of performance metrics including reliability, rate, and mode selection in overlaid D2D networks, which can provide invaluable insights that may be used to optimize the network design.
  \end{enumerate}

  The remainder of this paper is organized as follows. In Section II, we describe the system model and the fading models that will be used in this study. We introduce the interference of cellular and D2D networks in Section III, then propose a novel stochastic geometric approach for analyzing the D2D network performance in Section IV. Using the proposed approach, we evaluate the average rate and the bit error probability of D2D networks and present numerical results in Section V. Finally, Section VI concludes the paper with some closing remarks.

  \section{System Model}

  \subsection{Network Model}

  We consider a D2D network overlaid upon an uplink cellular network where a UE can directly communicate with other UEs without relying on the cellular infrastructure if a certain criterion is met. In overlay D2D networks, the cellular and D2D transmitters use orthogonal time/frequency resources by dividing the uplink spectrum into two non-overlapping portions. A spectrum partition factor $\beta$ is assigned for D2D communications and the remaining $1-\beta$ is allocated for cellular communications, where $0 \leq \beta \leq 1$. Overlaid spectrum access completely excludes cross-mode interference between cellular and D2D UEs, achieving a reliable link quality at the cost of lower spectrum utilization. 
   \begin{figure}[!t]
  	    \centering
  	        \includegraphics[width=0.5\linewidth]{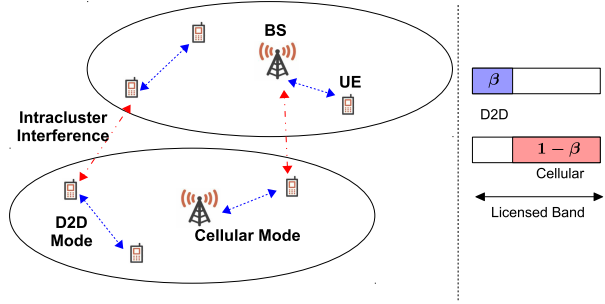}
  	    \caption{System Model and Overlay D2D Spectrum Sharing.}
  	    \label{fig.system1}
  \end{figure}

  \figref{fig.system1} depicts a high level overview of the system model where the locations of the nodes are modeled as homogeneous PPP in $\mathbb{R}^2$. The macro-cell BSs are uniformly distributed as PPP $\Psi$ with intensity $\lambda_b$ and the UEs are randomly deployed according to PPP $\Phi = \{ X_i \}$ with intensity $\lambda$, where $\Psi$ and $\Phi$ are independent point processes and $X_i$ denotes both the node and the coordinates of the $i$-th UE. Without loss of generality, we assume that the typical receiver is located at the origin, whether it is the cellular BS or D2D receiver. Based on the PPP $\Phi$, we define a marked PPP $\tilde{\Phi}$ as follows
  \begin{equation}
  	  \begin{split}
  	    \tilde{\Phi} = \{ X_i, \varrho_i, L_i, P_i\},
  	  \end{split}
  	  \label{eq_secII-1.001}
  \end{equation}
  where $\{ L_i \}$ and $\{ P_i \}$ denote the length of the link and transmit power of the $i$-th UE, respectively. $\{ \varrho_i \}$ is an indicator parameter for the types of the $i$-th UE which may be a potential D2D UE with probability $q = P(\varrho_i = 1)$, or a cellular UE with probability $1- q = P(\varrho_i = 0)$, where $q \in [0, 1]$.

  The received power $W$ from the $i$-th UE is $W = P_i L_i^{-\tau_i} G_i$, where $P_i$, $L_i$, $\tau_i > 2$,  and $G_i$ denote the transmit power, the link length, the path-loss exponent, and the small scale fading, respectively. The subscript $c$ and $d$ indicate the parameter for a cellular UE and potential D2D UE, respectively. In this study, we have isolated and focused on studying the impact of the small scale fading upon the system model proposed here. Nonetheless, it is worth highlighting that the model may be readily adapted to include shadowing using the same approach as utilized in \cite[Lemma 1]{Dhillon2014}.

  For the power control, we assume channel inversion over each link, so that the average received power is constant,  \textit{i.e.}, $P_i = {L_i}^{\tau_i}$, then $\mathbb{E}\left[W\right] = \mathbb{E}\left[G_i\right] = \bar{w}$. It is assumed that the channel coefficients $\{ G_i \}$ are independent of one another. Under this assumption, the received signal at the origin is written as
    \begin{equation}
    \begin{split}
      Y(t) = \sqrt{P_i L_i^{-\tau_i} G_i} S_i(t) + \sum_{j \in \mathcal{M}} \sqrt{P_j D_j^{-\tau_j} G_j} S_j(t) + Z(t),
    \end{split}
    \label{eq_secIII-1.001}
    \end{equation}
  where $t$ is the time index, $\mathcal{M}$ denotes the set of interfering UEs, $S_i(t)$ is the unit power signal over the intended link, $Z(t)$ is an additive white Gaussian noise process with noise power $N_0$, and the subscripts $i$ and $j$ represent the intended and interfering links, respectively. $P_j$, $D_j$, $\tau_j$, $G_j$ and $S_j(t)$ represent the transmit power, link length, path-loss exponent, small scale fading and unit power signal over the $j$-th interference link.
  Due to the power control and (\ref{eq_secIII-1.001}), the received SINR is given by
    \begin{equation}
    \begin{split}
      \mathrm{SINR} = \frac{W}{I + N_0} = \frac{G_i}{\sum_{j \in \mathcal{M}} P_j D_j^{-\tau_j} G_j + N_0}.
    \end{split}
    \label{eq_secIII-1.002}
    \end{equation}

  \subsection{Distributions of the Link Length}

  Based on the PPP assumption, the potential D2D UEs are randomly scattered over $\mathbb{R}^2$ and the internode distance between the closest D2D transmitter and receiver pair is governed by the Rayleigh distribution with a given D2D distance parameter $\xi > 0$ as follows \cite{Haenggi2005}
  \begin{equation}
    \begin{split}
      f_{L_d}(x) = 2 \pi \xi x \mathrm{e}^{-\xi \pi x^2}, \quad x \geq 0.
    \end{split}
    \label{eq_secII-1.002}
  \end{equation}

  For the cellular link, we use the approximate cellular uplink model proposed in \cite{Lin2014} where the coverage region of a macro-cell is approximated by a circular disk $\mathcal{A} = \mathcal{B}(0, R)$ with radius $R = \sqrt{\frac{1}{\pi \lambda_b}}$ and the active cellular transmitter is uniformly distributed in the cell range $\mathcal{A}$, where $\mathcal{B}(x, r)$ denotes a ball centered at $x$ with radius $r$. Consequently, the link distance between the cellular UE and the associated BS is given by
  \begin{equation}
    \begin{split}
      f_{L_c}(x) = \frac{2x}{R^2} = 2 \pi x \lambda_b, \quad 0 \leq x \leq R.
    \end{split}
    \label{eq_secII-1.003}
  \end{equation}
  The cellular links adopt orthogonal multiple access implying that only one uplink transmitter is active within each cell at a given time, whereas the medium access scheme for D2D is ALOHA with transmit probability $\varepsilon$ on each time slot, where $0 \leq \varepsilon \leq 1$.

  \subsection{Mode Selection and UE Classification}

    Each UE node $X_i \in \tilde{\Phi}$ in (\ref{eq_secII-1.001}) has an inherent type as indicated by the UE parameter $\varrho_i$, \textit{i.e.}, cellular UE ($\varrho_i = 0$) and potential D2D UE ($\varrho_i = 1$). The cellular UE always connects to the macro-cell BS, whereas the potential D2D UE may use either cellular or D2D mode depending on the associated receiver type. If the associated receiver is D2D, the potential D2D UE transmits in D2D mode. If the UE connects to a cellular BS, the potential D2D UE works in cellular mode.

    In this paper, we assume a distance-based mode selection scheme. That is, a potential D2D UE chooses the D2D mode if $L_d \leq \theta$, \textit{i.e.}, the D2D link length is not greater than a predefined mode selection threshold $\theta$, where the probability that a potential D2D UE chooses the D2D mode is given by
  \begin{equation}
    \begin{split}
     \mathrm{P}\left( L_d \leq \theta \right) = 1 - \mathrm{e}^{-\xi \pi \theta^2}.
    \end{split}
    \label{eq_secII-3.001-ext1}
  \end{equation}    
	Otherwise, cellular mode is selected.

Due to the mode selection, $\tilde{\Phi}$ can be divided into two non-overlapping point processes (PP) as follows
\begin{equation}
  \begin{alignedat}{3}
    &\text{UEs operating in cellular mode}:~   &&\Phi_c   &&~\text{with intensity } \lambda_c = \left[ (1-q) + q \mathrm{P}(L_d > \theta)\right] \lambda ,\\
    &\text{UEs operating in D2D mode}: \quad   &&\Phi_d   &&~\text{with intensity } \lambda_d = q \mathrm{P}(L_d \leq \theta) \lambda. 
    \label{eq_secII-3.002}
  \end{alignedat}
\end{equation}
However, the location of the interfering UEs in the cellular mode follows a location dependent thinning process due to the orthogonal multiple access, which makes the analysis intractable. Furthermore, the transmit power of the cellular UEs are correlated with the path-loss due to the power control policy. To achieve analytical tractability, we make the following assumptions. The accuracy of Assumptions 1-3 are validated through Section V. 

\begin{assumption}
The set of UEs operating in the D2D mode $\Phi_d$ constitute a PPP.	
\end{assumption}

\begin{assumption}
The set of UEs operating in the cellular mode $\Phi_c$ is approximated by a PPP $\hat{\Phi}_c$, where the active interfering UEs outside the cell coverage region $\mathcal{A}^c$ is distributed by a PPP $\Phi_{c, a}$ with intensity $\lambda_b = \frac{1}{\pi R^2}$ and the active cellular transmitter inside the coverage area $\mathcal{A}$ is uniformly distributed in a circular disk $\mathcal{B}(0, R)$ as (\ref{eq_secII-1.003})\footnote{
There are other ways to approximate cellular uplink. For instance, \cite{Singh2015} modeled the cellular uplink by an inhomogeneous PPP with appropriately thinned intensity.}.
\end{assumption}

\begin{assumption}
	$\Phi_d$ and $\hat{\Phi}_c$ are independent. 
\end{assumption}

  \subsection{Radio Channel Model}

  The physical channels of D2D, heterogeneous, or 5G networks are often characterized as inhomogeneous environments with clusters of  scattered waves \cite{Medbo2014}. For example, strong line-of-sight (LOS) components, correlated in-phase and quadrature scattered waves with unequal-power, and non-circular symmetry are frequently observed in the physical channel of 5G networks \cite{6953146}. Therefore, to evaluate the transmission performance over realistic radio channels for D2D networks, we adopt two general fading distributions which together can model both homogeneous and inhomogeneous radio environments.

    \subsubsection{$\kappa-\mu$ distribution}

    The $\kappa-\mu$ distribution represents the small-scale variation of the fading signal under LOS conditions, propagated through a homogeneous, linear, circularly symmetric environment \cite{Yacoub2007a, DaCosta2008, 6953146}. The $\kappa-\mu$ distribution is a general fading distribution that includes Rayleigh, Rician, Nakagami-\textit{m}, and One-sided Gaussian as special cases (See Table \ref{tab_channel1}).

    \begin{table}[!t]
    \centering
    \caption{Special Cases of the $\kappa-\mu$ and $\eta-\mu$ Fading Models.}
    \includegraphics[width=0.7\linewidth, height = 50mm]{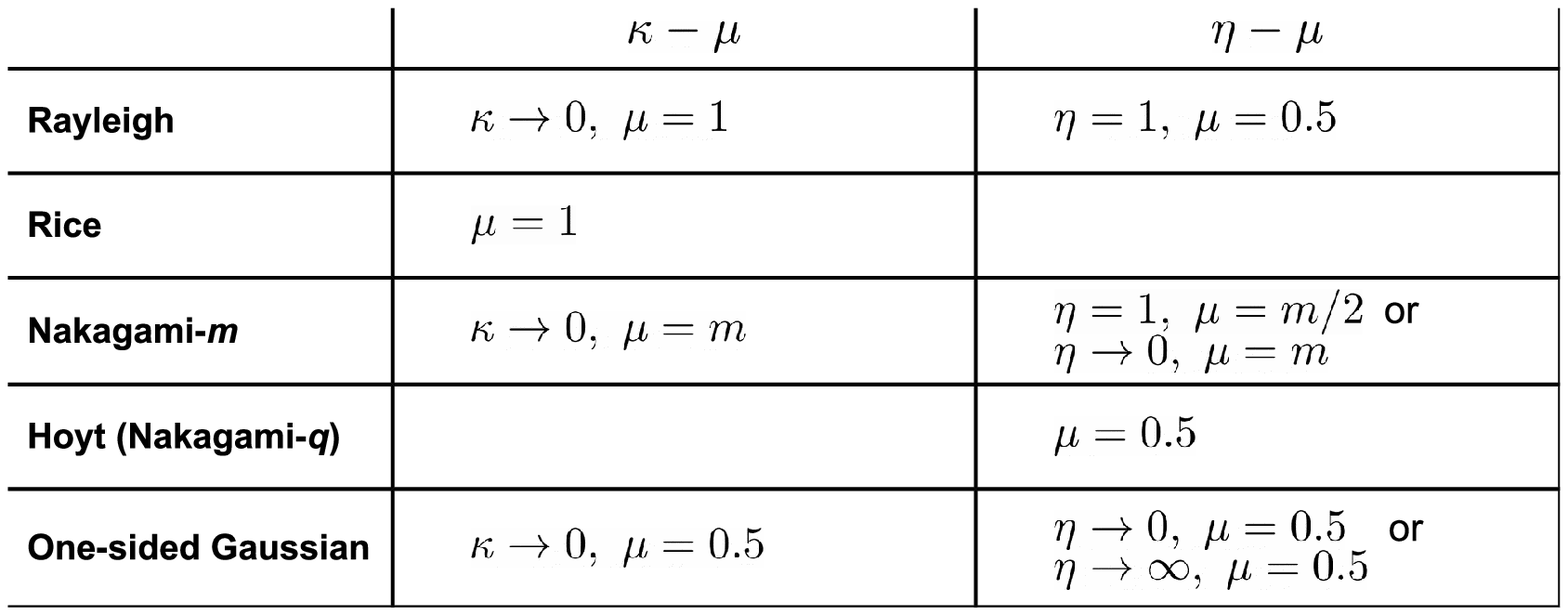}
    \label{tab_channel1}
  \end{table}

    The received signal envelope in a $\kappa-\mu$ fading channel consists of clusters of multi-path waves, where the signal within each cluster has an elective dominant component and scattered waves with identical powers. Following from this, the envelope $\sqrt{G}$ of a $\kappa-\mu$ fading signal can be written as
    \begin{equation}
  	  \begin{split}
  	    G = \sum_{j=1}^{n}\left( x_j^2 + y_j^2 \right),
  	  \end{split}
  	  \label{eq_secII-2.001}
  	  \end{equation}
  	  where $G$ is the small scale fading coefficient, $n$ is the number of multi-path clusters, and $x_j$ and $y_j$ are mutually independent Gaussian random variables with
  	  \begin{equation}
  	  \begin{split}
  	  \mathbb{E}(x_j) &= p_j, ~ \mathbb{E}(y_j) = q_j, \quad
  	  \mathbb{E}(x_j^2) = \mathbb{E}(y_j^2) = \sigma^2, \quad
  	  d^2 = \sum_{j=1}^{n}(p_j^2 + q_j^2).
  	  \end{split}
  	  \label{eq_secII-2.002}
    \end{equation}
  The physical meaning of $\kappa$ and $\mu$ parameters can be interpreted by (\ref{eq_secII-2.001}) and (\ref{eq_secII-2.002}): $\kappa = \frac{d^2}{2n\sigma^2}$ represents the ratio between the total power of the dominant components and the total power of the scattered waves, whereas $\mu$ is the real valued extension of $n$, \textit{i.e.} the number of multi-path clusters\footnote{ Note that $\mu$ as with $n$ is initially assumed to be a natural number, however for the $\kappa-\mu$ fading model, this restriction is relaxed to allow $\mu$ to assume any positive real value.}.

    The PDF, $j$-th moment and Laplace transform of $G$ are respectively given by \cite{Yacoub2007a, DaCosta2008}
  \begin{equation}
  	  \begin{split}
  	    f_{G}(x)
  	    &= \frac{ \mu }{\kappa^{\frac{\mu-1}{2}} \mathrm{e}^{\mu \kappa}}
  	    \left( \frac{1+\kappa}{\bar{w}} \right)^{\frac{\mu+1}{2}}
  	    x^{\frac{\mu-1}{2}}
  	    \exp\left( -\frac{\mu (1+\kappa)}{\bar{w}} x \right)
  	    \mathrm{I}_{\mu-1}\left( 2\mu \sqrt{ \frac{\kappa (1+\kappa)}{\bar{w}}x }\right), \\
  	    \mathbb{E}\left[ G^j \right] &= \frac{\bar{w}^j}{\mathrm{e}^{\mu \kappa} [(1+\kappa)\mu]^j}
  	    \frac{\Gamma(\mu+j)}{\Gamma(\mu)} \Hypergeometric{1}{1}{\mu+j}{\mu}{\mu \kappa}, \\
  	    \mathcal{L}_{G}(s) &= \mathbb{E}\left[ \exp(-s G) \right] =
  	    \frac{1}{\mathrm{e}^{\mu \kappa}} \sum_{n = 0}^{\infty} \frac{\left( \mu \kappa \right)^n}{n! \Gamma(n+\mu)}
  	    G_{1,1}^{1,1}\left( \frac{\mu (1+\kappa)}{s \bar{w}} \Bigg\vert { \mycom{1}{n+\mu} } \right),
  	  \end{split}
  	  \label{eq_secII-2.003}
    \end{equation}
    where $\bar{w} = \mathbb{E}[G]$, $\kappa$ and $\mu$ are positive real values, $\mathrm{I}_{\nu}(x)$ is the modified Bessel function of the first kind with order $\nu$, $\Gamma(t) = \int_{0}^{\infty}x^{t-1} \mathrm{e}^{-x} \mathrm{d}x $ is the Gamma function, $\Hypergeometric{1}{1}{a}{b}{x}$ is the confluent hypergeometric function, and $\MeijerG[a, b]{n}{p}{m}{q}{z}$ is the Meijer G-function (See Appendix I).

    \subsubsection{$\eta-\mu$ distribution}

    The $\eta-\mu$ distribution is used to represent small scale fading under non-line-of-sight (NLOS) conditions in inhomogeneous, linear, non-circularly symmetric environments \cite{Yacoub2007a, 4162456}. It is a general fading distribution that includes Hoyt (Nakagami-\textit{q}), One-sided Gaussian, Rayleigh, and Nakagami-\textit{m} as special cases (See Table \ref{tab_channel1}).

    The received signal in an $\eta-\mu$ distributed fading channel is composed of clusters of multi-path waves. The in-phase and quadrature components of the fading signal within each cluster are assumed to be either independent with unequal powers or correlated with identical powers. The envelope $\sqrt{G}$ of $\eta-\mu$ fading signal can be written in terms of the in-phase and quadrature components given by (\ref{eq_secII-2.001}), where $x_j$ and $y_j$ are mutually independent Gaussian random variables with
    \begin{equation}
  	  \begin{split}
  	  \mathbb{E}(x_j) &= \mathbb{E}(y_j) = 0, \quad
  	  \mathbb{E}(x_j^2) = \sigma_x^2, ~ \mathbb{E}(y_j^2) = \sigma_y^2,~ \sigma_x \neq \sigma_y.
  	  \end{split}
  	  \label{eq_secII-2.006}
    \end{equation}
  Physically, $\eta = \frac{\sigma_x^2}{\sigma_y^2}$ denotes the scattered-wave power ratio between the in-phase and quadrature components of each cluster of multi-path, and $\mu$ represents the real valued extension of $n/2$.

    The PDF, $j$-th moment and Laplace transform of $G$ are respectively given by \cite{Yacoub2007a, DaCosta2008}
    \begin{equation}
  	  \begin{split}
  	    f_{G}(x) &=
  	    \frac{ 2\sqrt{\pi} \mu^{\mu+\frac{1}{2}} h^{\mu} }{\Gamma(\mu) H^{\mu-\frac{1}{2}} \bar{w}^{\mu+\frac{1}{2}} }
  	    x^{\mu-\frac{1}{2}}
  	    \exp\left( -\frac{2 \mu h x}{\bar{w}} \right)
  	    \mathrm{I}_{\mu-\frac{1}{2}}\left( \frac{2 \mu H x}{\bar{w}}  \right),  \\
  	    \mathbb{E}\left[ G^j \right] &= \frac{\bar{w}^j}{h^{\mu+j} (2 \mu)^j}
  	    \frac{\Gamma(2\mu+j)}{\Gamma(2\mu)}
  	    \Hypergeometric{2}{1}{\mu+\frac{j}{2} +\frac{1}{2}, \mu+\frac{j}{2}}{\mu+\frac{1}{2}}{\left( \frac{H}{h}\right)^2}, \\
  	    \mathcal{L}_{G}(s) &=
  	    \frac{2\sqrt{\pi}}{\Gamma(\mu) h^{\mu}}
  	    \sum_{n = 0}^{\infty} \frac{2^{-2n-2\mu}}{n! \Gamma(n+\mu+\frac{1}{2})}
  	    \left( \frac{H}{h}\right)^{2n}
  	    G_{1,1}^{1,1}\left( \frac{2 \mu h}{s \bar{w}} \Bigg\vert { \mycom{1}{2n+2\mu} } \right),
  	  \end{split}
  	  \label{eq_secII-2.007}
    \end{equation}
    where $\bar{w} = \mathbb{E}[G]$, $\eta$ and $\mu$ are positive real values\footnote{
    $h$ and $H$ have two formats: format $1$ is $h = \frac{2 + \eta^{-1} + \eta}{4}$, $H = \frac{\eta^{-1} - \eta}{4}$ for $0 < \eta < \infty$, whereas format $2$ is $h = \frac{1}{1-\eta^2}$, $H = \frac{\eta}{1-\eta^2}$ for $-1 < \eta < 1$. In this paper, we will only consider format $1$ for notational simplicity.
    }, $h = \frac{2 + \eta^{-1} + \eta}{4}$, $H = \frac{\eta^{-1} - \eta}{4}$, and $\Hypergeometric{2}{1}{a}{b}{x}$ is the Gaussian hypergeometric function (See Appendix I).

  Since the Laplace transform of $G$ in (\ref{eq_secII-2.003}) and (\ref{eq_secII-2.007}) are both represented in terms of Meijer G-functions, computation and analysis for the $\kappa-\mu$ and $\eta-\mu$ distributions is quite complex. In the following lemma, we represent (\ref{eq_secII-2.003}) and (\ref{eq_secII-2.007}) in terms of some elementary functions, which are later used to simplify the D2D network performance evaluation.
  \begin{lemm}
  	   The Laplace transform of $G$ are given by
  	  \begin{equation}
  	  	  \mathcal{L}_{G}(s) =
  				\begin{dcases}
  				\frac{1}{\left( 1 + \frac{s \bar{w}}{\mu (1 + \kappa)}\right)^{\mu}}
  				  \exp\left( -\frac{\mu \kappa}{1 + \frac{\mu (1 + \kappa)}{s \bar{w}}}\right)
  				  &\mbox{for $\kappa-\mu$},\\
  			\frac{1}{h^{\mu}} \left[ \left( 1 + \frac{s \bar{w}}{2 \mu h}\right)^2 - \left( \frac{H}{h}\right)^2 \right]^{-\mu}
  				  &\mbox{for $\eta-\mu$},
  				\end{dcases}
  	\label{eq_secII-2.010}
  \end{equation}
  where $\bar{w} = \mathbb{E}[G]$, $h = \frac{2 + \eta^{-1} + \eta}{4}$, $H = \frac{\eta^{-1} - \eta}{4}$, and $\eta  > 0$.
  \end{lemm}

  \begin{proof}
  See Appendix III.
  \end{proof}

  \section{Interference Model of The Overlay D2D Network}

  In this section, we introduce the interference of cellular and D2D links and derive the Laplace transform of the interference for the generalized fading channels considered here.

  \subsection{D2D Mode}

  Let us consider a D2D link, where co-channel interference is generated by potential D2D UEs operating in D2D mode.
  Due to overlaid spectrum access, cross-mode interference between D2D and cellular links is excluded. Since the medium access scheme for D2D transmission is ALOHA with a transmit probability $\varepsilon$, the effective interference of the D2D link becomes a thinning PPP, denoted by $\varepsilon \Phi_d$, with intensity $\varepsilon \lambda_d$. Then, the interference at the intended D2D receiver is given by
    \begin{equation}
  	  \begin{split}
  	    I_d = \sum_{X_j \in \varepsilon \Phi_d \backslash \{ 0 \}}
  	    \hat{P}_{d, j} G_j {||X_j||}^{-\tau_d},
  	  \end{split}
  	  \label{eq_secIII-2.001}
    \end{equation}
  and the Laplace transform of $I_d$ is derived in the following lemma.

  \begin{lemm}
  		For overlay D2D, the Laplace transform of the interference at the D2D receiver is given by
  		  \begin{equation}
  		  \begin{split}
  		    \mathcal{L}_{I_d}(s) = \exp\left( -c ~ s^{\delta_d}\right),
  		  \end{split}
  		  \label{eq_secIII-2.002}
  		  \end{equation}
  		where the constant term $c$ for each channel distribution is derived as
  		  \begin{equation}
  		  \begin{split}
  		    c_{\kappa-\mu} &=
  		    \frac{q \varepsilon \lambda }{\hat{\xi}  \mathrm{e}^{\mu \kappa}} \cdot
  		    \frac{ \Hypergeometric{1}{1}{\mu+\delta_d}{\mu}{\mu \kappa}}{\mathrm{sinc}\left(\delta_d\right)} \cdot
  		    \left( \frac{\bar{w}}{(1 + \kappa) \mu} \right)^{\delta_d} \cdot
  		    \binom{\mu+\delta_d-1}{\delta_d},
  		      \end{split}
  		  \label{eq_secIII-2.003-a}
  		  \end{equation}
  		  \begin{equation}
  		  \begin{split}
  		    c_{\eta-\mu} &=
  		    \frac{q \varepsilon \lambda }{\hat{\xi}  h^{\mu}} \cdot
  		    \frac{ \Hypergeometric{2}{1}{\mu+\frac{\delta_d}{2} +\frac{1}{2}, \mu+\frac{\delta_d}{2}}{\mu + \frac{1}{2}}{\left( \frac{H}{h} \right)^2}}
  		    {\mathrm{sinc}\left(\delta_d\right)} \cdot
  		    \left( \frac{\bar{w}}{2 \mu h} \right)^{\delta_d} \cdot
  		    \binom{2\mu+\delta_d-1}{\delta_d},
  		  \end{split}
  		  \label{eq_secIII-2.003-b}
  		  \end{equation}
  		  for the $\kappa-\mu$ and $\eta-\mu$ distribution, respectively, with $\delta_d = \frac{2}{\tau_d}$, $\bar{w} = \mathbb{E}[G_i]$ and $\hat{\xi} = \frac{\xi}{\gamma\left( 2, \xi \pi \theta^2 \right)}$.
  \end{lemm}

  \begin{proof}
  See Appendix IV.
  \end{proof}

  \subsection{Cellular Mode}

    Over a cellular link, the interference originates from cellular UEs located outside the cell $\mathcal{A}$. Since the cellular interference is denoted by $\Phi_{c, a}$ and the area outside the cell is represented by $\mathcal{A}^c$, the interference at the intended cellular BS is written by
    \begin{equation}
  	  \begin{split}
  	  I_c = \sum_{X_j \in \Phi_{c, a} \cap \mathcal{A}^c}
  	    {P}_{c, j} G_j {||X_j||}^{-\tau_c},
  	  \end{split}
  	  \label{eq_secIII-3.001}
    \end{equation}
    and the Laplace transform of $I_c$ is given by the following lemma.

  \begin{lemm}
  		For overlay D2D, the Laplace transform of the interference at the cellular BS is given by
  		  \begin{equation}
  		  \begin{split}
  		    \mathcal{L}_{I_c}(s) = \exp\left( -2 \pi \lambda_b \int_{R}^{\infty} \left( 1 - \varphi(r) \right) r \mathrm{d}r  \right),
  		  \end{split}
  		  \label{eq_secIII-3.002}
  		  \end{equation}
  		where $\delta_c = \frac{2}{\tau_c}$ and $\varphi(r)$ for each channel distribution is written as
  		  \begin{equation}
  		  \begin{split}
  		    \varphi_{\kappa-\mu}(r) &=
  		    \frac{\pi \lambda_b r^2 \delta_c}{\mathrm{e}^{\mu \kappa}}
  		    \left( \frac{\mu (1 + \kappa)}{s \bar{w}} \right)^{\delta_c}
  		    \sum_{n = 0}^{\infty} \frac{(\mu \kappa)^n}{n! \Gamma\left(\mu + n\right)}
  		    G_{2,2}^{2,1}\left( \frac{\mu (1 + \kappa) \left( \pi \lambda_b \right)^{\frac{1}{\delta_c}} r^{\tau_c}}{s \bar{w}}
  		    \Bigg\vert { \mycom{1-\delta_c, 1}{0, n+\mu-\delta_c} } \right),
  		  \end{split}
  		  \label{eq_secIII-3.003-a}
  		  \end{equation}
  		  \begin{equation}
  		  \begin{split}
  		    \varphi_{\eta-\mu}(r) &=
  		    \frac{2 \pi^{\frac{3}{2}} \lambda_b r^2 \delta_c}{\Gamma(\mu) h^{\mu}}
  		    \left( \frac{2 \mu h}{s \bar{w}} \right)^{\delta_c}
  		    \sum_{n = 0}^{\infty}
  		    \frac{2^{-2n-2\mu}}{n! \Gamma\left(\mu + n + \frac{1}{2}\right)}
  		    \left( \frac{H}{h} \right)^{2 n} \\
  		    &\times
  		    G_{2,2}^{2,1}\left( \frac{2 \mu h \left( \pi \lambda_b \right)^{\frac{1}{\delta_c}} r^{\tau_c}}{s \bar{w}}
  		    \Bigg\vert { \mycom{1-\delta_c, 1}{0, 2n+2\mu-\delta_c} } \right),
  		  \end{split}
  		  \label{eq_secIII-3.003-b}
  		  \end{equation}
  		for the $\kappa-\mu$ and $\eta-\mu$ distribution, respectively.
   \end{lemm}

  \begin{proof}
  See Appendix V.
  \end{proof}

  \begin{remark}
  	We note a connection between Lemma 4 and other literatures on generalized fading. By using the conversion  method in \cite{Blaszczyszyn2013}, (\ref{eq_secIII-2.001}) can be expressed as $I_d = \sum_{X_j \in \varepsilon \Phi_d \backslash \{ 0 \}} \hat{P}_{d, j} {||Y_j||}^{-\tau_d}$,
  	where $Y_j = G_j^{-\frac{1}{\tau_d}} X_j$ for any $X_j \in  \varepsilon \Phi_d \backslash \{ 0 \}$. Due to \cite[Lemma 1]{Dhillon2014}, the new point process $\hat{\Phi} = \{ Y_j\}$ is also a homogeneous PPP with density $\varepsilon \lambda_d \mathbb{E}\left[ G^{\delta_d} \right]$. The Laplace transform of $I_d$ is given by
  			  \begin{equation}
  			  \begin{split}
  			    \mathcal{L}_{I_d}(s) = \exp\left( -\pi \varepsilon \lambda_d  \mathbb{E}\left[ G^{\delta_d} \right] \mathbb{E}\left[ {\hat{P}_d}^{\delta_d} \right] \Gamma\left( 1-\delta_d \right) s^{\delta_d}\right),
  			  \end{split}
  			  \label{eq_secIII-2.005}
  			  \end{equation}
  	where we used similar derivation as Appendix IV. Then, by substituting $\lambda_d = q \mathrm{P}(L_d \leq \theta) \lambda$, (\ref{eq_secII-2.003}), (\ref{eq_secIII-1.003}) into (\ref{eq_secIII-2.005}), we get $c = \pi \varepsilon \lambda_d  \mathbb{E}\left[ G^{\delta_d} \right] \mathbb{E}\left[ {\hat{P}_d}^{\delta_d} \right] \Gamma\left( 1-\delta_d \right)$, which corresponds to (\ref{eq_secIII-2.003-a}) for the $\kappa-\mu$ and (\ref{eq_secIII-2.003-b}) for the $\eta-\mu$ distribution. The same method can not be applied for Lemma 5, since the conversion method is not applicable when there is an exclusion zone in the interference field.
  \end{remark}

  \section{Stochastic Geometric Framework for System Performance Evaluation}

  To evaluate the network performance,  one normally needs to calculate the average of some functions of the SINR for a given SINR distribution $f_{\gamma}(x)$ where SINR $\gamma$ is defined in (\ref{eq_secIII-1.002}). The average of an arbitrary function of the SINR represents most commonly used characteristics, such as the spectral efficiency, error probability, statistical moments, outage or coverage probability, etc. Quite often this can be an extremely challenging task due in part to the complex nature of functions involved and also because, within the stochastic geometry framework, the closed form expression of the SINR $f_{\gamma}(x)$ distribution is known only for some special cases, such as Rayleigh \cite{Dhillon2012} or Nakagami-\textit{m} fading \cite{Tanbourgi2014}. Instead, for many cases, we can evaluate the Laplace transform of the interference $\mathcal{L}_{I}(s)$ and the PDF of the intended channel $f_{W}(x)$.

  To this end, an analytical method was proposed by Hamdi in \cite{Hamdi2007} to compute $\mathbb{E}\left[ g\left( \gamma \right)\right]$ for an arbitrary function of the SINR $g(\gamma)$ over Nakagami-\textit{m} fading channel. In the following theorems, we now generalize this method within the stochastic geometry framework and evaluate the average of an arbitrary function of the SINR by using $\mathcal{L}_{I}(s)$ and $f_{W}(x)$ only, without knowing $f_{\gamma}(x)$. Theorem $1$ and $2$ assumes that the received signal envelope of the intended link $\sqrt{G_i} = \sqrt{W}$ undergoes $\kappa-\mu$ and $\eta-\mu$ fading, respectively.

  \begin{thm}
  We assume that the received signal envelope of the intended link $\sqrt{W}$ is a $\kappa-\mu$ distributed random variable and $I$ is an arbitrary random variable that is independent of $W$. Then, the average $\mathbb{E}\left[ g\left( \frac{W}{I+N_0}\right)\right]$ is given by
    \begin{equation}
    \begin{split}
      &\mathbb{E}\left[ g\left( \frac{W}{I+N_0}\right)\right] =
      g(0) + \sum_{n = 0}^{\infty} \frac{\left( \mu \kappa \right)^n}{n! ~\mathrm{e}^{\mu \kappa}}
      \int_{0}^{\infty} g_{\mu+n}\left( z\right)
      \mathcal{L}_{I}\left( \frac{\mu(1+\kappa)}{\bar{w}}z\right)
      \mathrm{e}^{-\frac{\mu(1+\kappa)}{\bar{w}}z} \mathrm{d}z\\
      &= g(0) + \frac{\bar{w}}{\mu(1+\kappa) N_0} \sum_{m=1}^{M}
      \sum_{n = 0}^{\infty} \frac{c_m \left( \mu \kappa \right)^n}{n! ~ \mathrm{e}^{\mu \kappa}}
      g_{\mu+n}\left( \frac{\bar{w} x_m}{\mu(1+\kappa) N_0} \right)
      \mathcal{L}_{I}\left( \frac{x_m}{N_0}\right) + {R}_M,
    \end{split}
    \label{eq_secIV-3.001}
    \end{equation}
  where $g(x)$ is an analytic function, $\kappa$ and $\mu$ are non-negative real valued constants, $\bar{w} = \mathbb{E}[W]$, and
    \begin{equation}
    \begin{split}
      R_M &= \frac{\bar{w}}{\mu(1+\kappa) N_0}
      \sum_{m=M+1}^{\infty}
      \sum_{n = 0}^{\infty} \frac{c_m \left( \mu \kappa \right)^n}{n! ~ \mathrm{e}^{\mu \kappa}}
      g_{\mu+n}\left( \frac{\bar{w} x_m}{\mu(1+\kappa) N_0} \right)
      \mathcal{L}_{I}\left( \frac{x_m}{N_0}\right),\\
      g_{i}(z) &= \frac{1}{\Gamma(\mu+n)} \frac{\mathrm{d}^i}{\mathrm{d}z^i}\left( z^{\mu+n-1} g(z) \right),
    \end{split}
    \label{eq_secIV-3.002}
    \end{equation}
  $c_m$ and $x_m$ are the $m$-th weight and abscissa of the $M$-th order Laguerre polynomial, respectively.
  \end{thm}

  \begin{proof}
  See Appendix VI.
  \end{proof}

  \begin{thm}
  We assume that the received signal envelope of the intended link $\sqrt{W}$ is a $\eta-\mu$ distributed random variable and $I$ denote an arbitrary random variable such that $W$ and $I$ are to be assumed independent. Then, $\mathbb{E}\left[ g\left( \frac{W}{I+N_0}\right)\right]$ can be expressed as
    \begin{equation}
    \begin{split}
      &\mathbb{E}\left[ g\left( \frac{W}{I+N_0}\right)\right] = g(0) \sum_{n=0}^{\infty} a_n
      + \sum_{n=0}^{\infty} a_n \int_{0}^{\infty}
      g_{2\mu+2n}\left( z\right)
      \mathcal{L}_{I}\left( \frac{2\mu h}{\bar{w}}z\right)
      \mathrm{e}^{-\frac{2\mu h N_0}{\bar{w}}z} \mathrm{d}z\\
      &= g(0)
      + \frac{\bar{w}}{2\mu h N_0}
      \sum_{m = 1}^{M}
      \sum_{n=0}^{\infty} a_n c_m g_{2\mu+2n}\left( \frac{\bar{w} x_m}{2\mu h N_0} \right)
      \mathcal{L}_{I}\left( \frac{x_m}{N_0} \right) + R_M,
    \end{split}
    \label{eq_secIV-3.003}
    \end{equation}
  where $\eta$ and $\mu$ are non-negative real valued constants, $h = \frac{2 + \eta^{-1} + \eta}{4}$, $H = \frac{\eta^{-1} - \eta}{4}$,
  $\bar{w} = \mathbb{E}[W]$, and
    \begin{equation}
    \begin{split}
      a_n &= \binom{n+\mu-1}{n} \frac{H^{2n}}{h^{\mu + 2n}}, \quad
      g_{j}(z) = \frac{1}{\Gamma(2\mu+2n)} \frac{\mathrm{d}^j}{\mathrm{d}z^j}\left( z^{2\mu+2n-1} g(z) \right),\\
      R_M &= \frac{\bar{w}}{2\mu h N_0}
      \sum_{m = M+1}^{\infty}
      \sum_{n=0}^{\infty} a_n c_m g_{2\mu+2n}\left( \frac{\bar{w} x_m}{2\mu h N_0} \right)
      \mathcal{L}_{I}\left( \frac{x_m}{N_0} \right),
    \end{split}
    \label{eq_secIV-3.004}
    \end{equation}
  $c_m$ and $x_m$ are the $m$-th weight and abscissa of the $M$-th order Laguerre polynomial, respectively.
  \end{thm}

  \begin{proof}
  The proof of Theorem 2 is similar to that Theorem 1, where we used the following relations
    \begin{equation}
    \begin{split}
      a_n = &\frac{\sqrt{\pi}}{\Gamma(\mu) 2^{2\mu-1} h^{\mu}}
      \frac{\Gamma(2\mu+2n)}{\Gamma(n+1) \Gamma\left(\mu+n+\frac{1}{2}\right)} \left( \frac{H}{2 h} \right)^{2n}
      = \binom{n+\mu-1}{n} \frac{H^{2n}}{h^{\mu + 2n}},\\
      &\sum_{n = 0}^{\infty} a_n = \frac{1}{h^{\mu}} \sum_{n = 0}^{\infty} \left( -1 \right)^n
      \binom{-\mu}{n} \frac{H^{2n}}{h^{2n}} = \frac{1}{\left( h - \frac{H^2}{h} \right)^{\mu}} = 1,
    \end{split}
    \label{eq_app_IX-ext-001}
    \end{equation}
  by using (\ref{eq_app_III-007}) in the first expression and
  applying (\ref{eq_app_III-010}) with $H^2 = h(h-1)$ in the second expression.
  \end{proof}

  We note that Theorems 1 and 2 make no assumption on the underlying distribution of the constituent interference channels. Therefore they can be applied even when the intended and interfering links are described by different fading models.

  \section{Performance Evaluation and Numerical Results}

  In this section, we apply Theorems 1 and 2 to an overlaid D2D network. We evaluate the average rate and bit error probability of the cellular and D2D UEs and compare their performance using a series of numerical evaluations.

  \subsection{Average Rate}

  The transmission rate of an overlaid D2D network is determined in part by the spectrum partition factor $\beta$. Since $\beta$ is the fraction of the available spectrum allocated for D2D transmission, the rate of the potential D2D UEs operating in the D2D mode is $\hat{R}_d = \beta C_d$, where $C_d$ denotes the spectral efficiency of a D2D link. Since the D2D mode uses an ALOHA medium access strategy with transmit probability $\varepsilon$, the spectral efficiency of a D2D link is given by
    \begin{equation}
    \begin{split}
      C_d = \varepsilon \mathbb{E}\left[ \log\left( 1 + \frac{W}{I_d + N_0}\right) \right].
    \end{split}
    \label{eq_secV-1.001}
    \end{equation}

  Similarly, the rate of a cellular UE is ${R}_c = \left( 1- \beta \right) C_c$, where $C_c$ represents the spectral efficiency of a cellular link. Due to the orthogonal multiple access, only one cellular UE accesses the cellular link at a time and the spectral efficiency of a cellular link is consequently given by
    \begin{equation}
    \begin{split}
      C_c = \mathbb{E}\left[ \frac{1}{N} \right] \mathbb{E}\left[ \log\left( 1 + \frac{W}{I_c + N_0}\right) \right]
      = \frac{\lambda_b}{\lambda_c}\left( 1 - \mathrm{e}^{-\frac{\lambda_b}{\lambda_c}}\right)
      \mathbb{E}\left[ \log\left( 1 + \frac{W}{I_c + N_0}\right) \right],
    \end{split}
    \label{eq_secV-1.002}
    \end{equation}
  where $N$ is the number of potential cellular UEs within the cell $\mathcal{A}$ and the average of $1/N$ is evaluated in \cite{Lin2014} as
  $\mathbb{E}\left[ \frac{1}{N} \right] = \frac{\lambda_b}{\lambda_c}\left( 1 - \mathrm{e}^{-\frac{\lambda_c}{\lambda_b}}\right)$.

  Based on the mode selection scheme, the potential D2D UE may choose either cellular or D2D mode. If D2D mode is selected, the rate of potential D2D UE $R_d$ is given by $R_d = \hat{R}_d$, whereas if cellular mode is selected, $R_d = R_c$. Hence, the average rates of a potential D2D UE $R_d$ can be calculated by using total probability as shown in the following theorem.

  \begin{thm}
   For an overlaid D2D network, the average rates of a cellular UE $R_c$, a potential D2D UE operating in D2D mode $\hat{R}_d$, and a potential D2D UE $R_d$ are given by
    \begin{equation}
    \begin{split}
      R_c &= \left( 1- \beta \right) \frac{\lambda_b}{\lambda_c}\left( 1 - \mathrm{e}^{-\frac{\lambda_c}{\lambda_b}}\right)
      \mathbb{E}\left[ \log\left( 1 + \frac{W}{I_c + N_0}\right) \right],\\
      \hat{R}_d &= \beta \varepsilon \mathbb{E}\left[ \log\left( 1 + \frac{W}{I_d + N_0}\right) \right],\\
      R_d &= R_c \mathrm{P}\left( L_d > \theta \right) + \hat{R}_d \mathrm{P}\left( L_d \leq \theta \right),
    \end{split}
    \label{eq_secV-1.003}
    \end{equation}
  where $\beta \in \left[ 0, 1\right]$ is the spectrum partition factor and $\mathrm{P}\left( L_d \leq \theta\right)$ as calculated in (\ref{eq_secII-3.001-ext1}). The average of the logarithm function can be evaluated using Theorems 1, and 2 as
  \begin{equation}
  \begin{split}
  \mathbb{E}\left[ \log\left( 1 + \frac{W}{I + N_0}\right) \right]
  &=
  \frac{\bar{w}}{\mu(1+\kappa) N_0} \sum_{n = 0}^{\infty}
  \frac{\left( \mu \kappa \right)^n}{n! ~ \mathrm{e}^{\mu \kappa}}
  \int_{0}^{\infty}
  g_{\mu+n}\left( \frac{\bar{w} x}{\mu(1+\kappa) N_0} \right)
  \mathcal{L}_{I}\left( \frac{x}{N_0}\right)
  \mathrm{e}^{-x}
  \mathrm{d}x\\
  &\simeq
  \frac{\bar{w}}{\mu(1+\kappa) N_0} \sum_{m=1}^{M}
  \sum_{n = 0}^{\infty} \frac{c_m \left( \mu \kappa \right)^n}{n! ~ \mathrm{e}^{\mu \kappa}}
  g_{\mu+n}\left( \frac{\bar{w} x_m}{\mu(1+\kappa) N_0} \right)
  \mathcal{L}_{I}\left( \frac{x_m}{N_0}\right),
  \end{split}
  \label{eq_secV-1.004-a}
  \end{equation}
  for a $\kappa-\mu$ distributed signal envelope $\sqrt{G_i} = \sqrt{W}$, and
  \begin{equation}
  \begin{split}
  \mathbb{E}\left[ \log\left( 1 + \frac{W}{I + N_0}\right) \right]
  &=
  \frac{\bar{w}}{2 \mu h N_0} \sum_{n = 0}^{\infty}
  a_n
  \int_{0}^{\infty}
  g_{2 \mu+ 2 n}\left( \frac{\bar{w} x}{2 \mu h N_0} \right)
  \mathcal{L}_{I}\left( \frac{x}{N_0}\right)
  \mathrm{e}^{-x}
  \mathrm{d}x\\
  &\simeq
  \frac{\bar{w}}{2\mu h N_0}
  \sum_{m = 1}^{M}
  \sum_{n=0}^{\infty} a_n c_m g_{2\mu+2n}\left( \frac{\bar{w} x_m}{2\mu h N_0} \right)
  \mathcal{L}_{I}\left( \frac{x_m}{N_0} \right),
  \end{split}
  \label{eq_secV-1.004-b}
  \end{equation}
  for a $\eta-\mu$ distributed signal envelope, where $a_n$ is defined in (\ref{eq_secIV-3.004}), $c_m$ and $x_m$ are the $m$-th weight and abscissa of the $M$-th order Laguerre polynomial, respectively. Note that we have omitted $R_M$ in (\ref{eq_secIV-3.001}), and (\ref{eq_secIV-3.003}) since it rapidly converges to zero \cite{Gradshteyn1994}. The Laplace transforms of the interference $\mathcal{L}_{I}(s)$ for D2D and cellular link are derived in (\ref{eq_secIII-2.002}) and (\ref{eq_secIII-3.002}), respectively. The derivative terms $g_i(z)$ in (\ref{eq_secV-1.004-a}) and (\ref{eq_secV-1.004-b}) are of the form $g_i(z) = \frac{1}{z} \left( 1 - \frac{1}{(1+z)^i}\right)$  following from \cite[eq. 18]{Hamdi2007}.
  \end{thm}

  In the following, we consider the special cases of the $\kappa-\mu$ and $\eta-\mu$ fading models and compare the rate for different fading conditions.
  \subsubsection{Nakagami-\textit{m}} As indicated in Table \ref{tab_channel1}, Nakagami-\textit{m} fading can be obtained from $\kappa-\mu$ fading by setting $\kappa \rightarrow 0$, $\mu = m$ or from $\eta-\mu$ by setting either $\eta = 1, \mu = m/2$ or $\eta \rightarrow 0, \mu = m$. The average rates of an overlaid D2D networks are given in (\ref{eq_secV-1.003}), where the average of the logarithm can be simplified for Nakagami-\textit{m} as follows
   \begin{equation}
    \begin{split}
    \mathbb{E}\left[ \log\left( 1 + \frac{W}{I+N_0}\right)\right] &=
    \frac{\bar{w}}{m N_0}
      \int_{0}^{\infty} g_{m}\left( \frac{\bar{w} x}{m N_0} \right)
      \mathcal{L}_{I}\left( \frac{x}{N_0} \right)
      \mathrm{e}^{-x} \mathrm{d}x,
    \end{split}
    \label{eq_secV-1.014}
    \end{equation}
  where $g_{m}(z) = \frac{1}{z} \left( 1 - \frac{1}{(1+z)^m}\right)$ and a detailed proof of (\ref{eq_secV-1.014}) is provided in Appendix VII.   For the D2D and cellular links, the Laplace transform $\mathcal{L}_{I}(s)$ in (\ref{eq_secV-1.014}) is respectively given as below
   \begin{equation}
    \begin{split}
    \mathcal{L}_{I_d}(s) &= \exp\left( -    \frac{q \varepsilon \lambda }{\hat{\xi} \mathrm{sinc}\left(\delta_d\right)} \cdot
      \left( \frac{\bar{w}}{m} \right)^{\delta_d} \cdot
      \binom{m+\delta_d-1}{\delta_d}
       s^{\delta_d}\right),\\
    \mathcal{L}_{I_c}(s) &=
    \exp\left( -\frac{2}{R^2} \int_{R}^{\infty} \left( 1 -
  \Hypergeometric{2}{1}{m, \delta_c}{1 + \delta_c}{-\frac{s R^{\tau_c} {r}^{-\tau_c} \bar{w}}{m}}
   \right) r \mathrm{d}r  \right).
    \end{split}
    \label{eq_secV-1.015}
    \end{equation}

  \subsubsection{Rayleigh} Rayleigh fading can be obtained from $\kappa-\mu$ fading by setting $\kappa \rightarrow 0$, $\mu = 1$ or from $\eta-\mu$ by setting $\eta = 1, \mu =0.5$. Then, by substituting $m = 1$ in (\ref{eq_secV-1.014}), the average of the logarithm can be simplified for Rayleigh as follows
   \begin{equation}
    \begin{split}
    \resizebox{0.91\hsize}{!}{%
    $\mathbb{E}\left[ \log\left( 1 + \frac{W}{I+N_0}\right)\right] =
    \frac{\bar{w}}{\mu N_0}
      \int_{0}^{\infty} g^{'}\left( \frac{\bar{w} x}{\mu N_0} \right)
      \mathcal{L}_{I}\left( \frac{x}{N_0} \right)
      \mathrm{e}^{-x} \mathrm{d}x
     =  \int_{0}^{\infty}
      \mathcal{L}_{I}\left( \frac{\mathrm{e}^{t}-1}{\bar{w}} \right)
      \mathrm{e}^{-\frac{(\mathrm{e}^{t}-1) N_0}{\bar{w}}} \mathrm{d}t,$
      }
    \end{split}
    \label{eq_secV-1.012}
    \end{equation}
    where we used a change of variables $t = \log\left( 1 + \frac{\bar{w}}{N_0}x \right)$ with
    $g^{'} = (1 + x)^{-1}$ to achieve the second equality of (\ref{eq_secV-1.012}).
    The last expression in (\ref{eq_secV-1.012}) matches the well-known result for Rayleigh fading in \cite{Haenggi2013}, validating our generalized approach in (\ref{eq_secIV-3.001}) and (\ref{eq_secIV-3.003}).
    For the D2D and cellular links, the Laplace transform $\mathcal{L}_{I}(s)$ in (\ref{eq_secV-1.012}) is respectively given as below by substituting $m = 1$ in (\ref{eq_secV-1.015}),
   \begin{equation}
    \begin{split}
    \mathcal{L}_{I_d}(s) &= \exp\left( -\frac{q \varepsilon \lambda \bar{w}^{\delta_d}}{\hat{\xi} \mathrm{sinc}\left(\delta_d\right)} s^{\delta_d}\right),\\
    \mathcal{L}_{I_c}(s) &=
    \exp\left( -\frac{2}{R^2} \int_{R}^{\infty} \left( 1 -
  \Hypergeometric{2}{1}{1, \delta_c}{1 + \delta_c}{-{s R^{\tau_c} {r}^{-\tau_c} \bar{w}}} \right) r \mathrm{d}r  \right).
    \end{split}
    \label{eq_secV-1.013}
    \end{equation}

  \subsubsection{One-sided Gaussian} One-sided Gaussian fading can be obtained from $\kappa-\mu$ fading by setting $\kappa \rightarrow 0$, $\mu = 0.5$ or from $\eta-\mu$ by setting $\eta \rightarrow 0$, or $\eta \rightarrow \infty$ with $\mu = 0.5$. The average of the logarithm in (\ref{eq_secV-1.003}) can be simplified for One-sided Gaussian by setting $\mu = 0.5$ in (\ref{eq_secV-1.014}) as follows
   \begin{equation}
    \begin{split}
    \mathbb{E}\left[ \log\left(1 + \frac{W}{I+N_0}\right)\right] &=
    \frac{2 \bar{w}}{N_0}
      \int_{0}^{\infty} g_{0.5}\left( \frac{2\bar{w} x}{N_0} \right)
      \mathcal{L}_{I}\left( \frac{x}{N_0} \right)
      \mathrm{e}^{-x} \mathrm{d}x,
    \end{split}
    \label{eq_secV-1.022}
    \end{equation}
  and the Laplace transform $\mathcal{L}_{I}(s)$ in (\ref{eq_secV-1.022}) is given by
   \begin{equation}
    \begin{split}
    \mathcal{L}_{I_d}(s) &= \exp\left( -
    \frac{q \varepsilon \lambda \left( 2 \bar{w}\right)^{\delta_d}}{\hat{\xi} \mathrm{sinc}\left(\delta_d\right)}
       \cdot
      \frac{\Gamma(\delta_d + 0.5)}{\sqrt{\pi} \Gamma(\delta_d + 1)}
       s^{\delta_d}\right),\\
    \mathcal{L}_{I_c}(s) &=
    \exp\left( -\frac{2}{R^2} \int_{R}^{\infty} \left( 1 -
  \Hypergeometric{2}{1}{\frac{1}{2}, \delta_c}{1 + \delta_c}{-{2 s R^{\tau_c} {r}^{-\tau_c} \bar{w}}}
   \right) r \mathrm{d}r  \right),
    \end{split}
    \label{eq_secV-1.023}
    \end{equation}
    for the D2D and cellular links, respectively,
    where we applied (\ref{eq_app_III-007}).

  \subsubsection{Rician} Rician fading can be obtained from $\kappa-\mu$ fading by setting $\mu = 1$.
  The average of the logarithm in (\ref{eq_secV-1.003}) is derived by setting $\mu = 1$ in (\ref{eq_secV-1.004-a}) as follows
   \begin{equation}
    \begin{split}
    \mathbb{E}\left[ \log\left( 1 + \frac{W}{I+N_0}\right)\right] &=
    \sum_{n = 0}^{\infty} \frac{\left( \kappa \right)^n}{n! ~\mathrm{e}^{\kappa}}
      \int_{0}^{\infty} g_{1+n}\left( z\right)
      \mathcal{L}_{I}\left( \frac{1+\kappa}{\bar{w}}z\right)
      \mathrm{e}^{-\frac{(1+\kappa)}{\bar{w}}z} \mathrm{d}z,
    \end{split}
    \label{eq_secV-1.016}
    \end{equation}
  and the Laplace transform $\mathcal{L}_{I}(s)$ in (\ref{eq_secV-1.016}) is given by
   \begin{equation}
    \begin{split}
    \mathcal{L}_{I_d}(s) &= \exp\left( -
    \frac{q \varepsilon \lambda }{\hat{\xi} \mathrm{e}^{\kappa}} \cdot
     \frac{ \Hypergeometric{1}{1}{1+\delta_d}{1}{\kappa}}{\mathrm{sinc}\left(\delta_d\right)}
     \cdot
     \left( \frac{\bar{w}}{1 + \kappa} \right)^{\delta_d}
       s^{\delta_d}\right),\\
    \mathcal{L}_{I_c}(s) &=
    \exp\left( -\frac{2}{R^2} \int_{R}^{\infty} \left( 1 -
   \varphi(r)
   \right) r \mathrm{d}r  \right),
   \end{split}
    \label{eq_secV-1.017}
    \end{equation}
    for the D2D and cellular links, respectively,
    where
   \begin{equation}
    \begin{split}
  \varphi(r) &=
  \frac{\delta_c}{\mathrm{e}^{\kappa}}\left( \frac{r}{R}\right)^2
  \left( \frac{1 + \kappa}{s \bar{w}} \right)^{\delta_c}
  \sum_{n = 0}^{\infty} \frac{\kappa^n}{n! n!}
  G_{2,2}^{2,1}\left( \frac{1 + \kappa}{s \bar{w}}\left( \frac{r}{R}\right)^{\tau_c}
  \Bigg\vert { \mycom{1-\delta_c, 1}{0, n+1-\delta_c} } \right).
   \end{split}
    \label{eq_secV-1.018}
    \end{equation}

  \subsubsection{Hoyt (Nakagami-\textit{q})} Hoyt fading is a special case of $\eta-\mu$ fading which is obtained by setting $\mu =0.5$. Then, by applying (\ref{eq_app_III-007}) , the average of the logarithm in (\ref{eq_secV-1.003}) is simplified as follows
   \begin{equation}
    \begin{split}
    \mathbb{E}\left[ \log\left( 1 + \frac{W}{I+N_0}\right)\right] &=
    \sum_{n = 0}^{\infty} \frac{\Gamma(n+0.5)}{\sqrt{\pi} ~ n!} \frac{\bar{w} H^{2n}}{N_0 h^{2n + 1.5}}
    \int_{0}^{\infty} g_{2n + 1}\left( \frac{\bar{w} x}{h N_0} \right)
    \mathcal{L}_{I}\left( \frac{x}{N_0} \right) \mathrm{e}^{-x} \mathrm{d}x,
    \end{split}
    \label{eq_secV-1.019}
    \end{equation}
  and the Laplace transform $\mathcal{L}_{I}(s)$ in (\ref{eq_secV-1.019}) is given by
   \begin{equation}
    \begin{split}
    \mathcal{L}_{I_d}(s) &= \exp\left( -
  	\frac{q \varepsilon \lambda }{\hat{\xi} ~h^{0.5} ~\mathrm{sinc}\left(\delta_d\right)} \cdot
  	\Hypergeometric{2}{1}{1+\frac{\delta_d}{2}, \frac{1}{2} +
      \frac{\delta_d}{2}}{1}{\left(\frac{H}{h}\right)^2}
     \cdot
     \left( \frac{\bar{w}}{1 + \kappa} \right)^{\delta_d}
       s^{\delta_d}
  	\right),\\
    \mathcal{L}_{I_c}(s) &=
    \exp\left( -\frac{2}{R^2} \int_{R}^{\infty} \left( 1 -
   \varphi(r)
   \right) r \mathrm{d}r  \right),
   \end{split}
    \label{eq_secV-1.020}
    \end{equation}
    for the D2D and cellular links, respectively,
    where
   \begin{equation}
    \begin{split}
  \varphi(r) &=
  \frac{\delta_c}{h^{0.5}}\left( \frac{r}{R}\right)^2
  \left( \frac{h}{s \bar{w}} \right)^{\delta_c}
  \sum_{n = 0}^{\infty} \frac{2^{-2n}}{n! n!} \left( \frac{H}{h}\right)^{2n}
  G_{2,2}^{2,1}\left( \frac{h}{s \bar{w}}
  \left( \frac{r}{R}\right)^{\tau_c}
  \Bigg\vert { \mycom{1-\delta_c, 1}{0, 2n+1-\delta_c} } \right).
   \end{split}
    \label{eq_secV-1.021}
    \end{equation}

  The rate analysis in (\ref{eq_secV-1.012})-(\ref{eq_secV-1.023}) are based on the common assumption
  of an identical fading distribution across the intended and interfering links.
  We note that Theorem 3 can be applied to the general case when different fading distributions affect
  the intended and interfering links, as indicated in Table \ref{tab_II}.
  For example, if the received signal envelope of the intended link is $\kappa-\mu$ distributed and
  that of the interference link is $\eta-\mu$ distributed, then the average of the
  logarithm in (\ref{eq_secV-1.004-a}) can be evaluated by using the Laplace transform $\mathcal{L}_{I}(s)$
  in (\ref{eq_secIII-2.002}) and (\ref{eq_secIII-2.003-b}) for D2D link, or
  (\ref{eq_secIII-3.002}) and (\ref{eq_secIII-3.003-b}) for cellular link, respectively.

  \begin{table}[!t]
  \centering
  \caption{Performance Evaluation over Multiple Fading Combinations.}
  \label{tab_II}
  \resizebox{\textwidth}{!}{%
  \begin{tabular}{c|c|c|c|c|c}
  \hline
  Intended Link  & Interference Link
  & $\mathbb{E}\left[ \log\left( 1 + \gamma\right)\right]$
  & $\mathbb{E}\left[ {\Gamma\left(b, a\gamma \right)}/{2 \Gamma(b)} \right]$
  & $\mathcal{L}_{I_d}(s)$: D2D Link   & $\mathcal{L}_{I_c}(s)$: Cellular Link   \\ \hline
  $\kappa-\mu$ & $\kappa-\mu$ & \multirow{2}{*}{(\ref{eq_secV-1.004-a})}
  & \multirow{2}{*}{(\ref{eq_secV-2.003-a})}
  & (\ref{eq_secIII-2.002}) \& (\ref{eq_secIII-2.003-a})
  & (\ref{eq_secIII-3.002}) \& (\ref{eq_secIII-3.003-a}) \\ \cline{1-2} \cline{5-6}
  $\kappa-\mu$ & $\eta-\mu$ &                     &
  & (\ref{eq_secIII-2.002}) \& (\ref{eq_secIII-2.003-b})
  & (\ref{eq_secIII-3.002}) \& (\ref{eq_secIII-3.003-b}) \\ \hline
  $\eta-\mu$ & $\kappa-\mu$ & \multirow{2}{*}{(\ref{eq_secV-1.004-b})}
  & \multirow{2}{*}{(\ref{eq_secV-2.003-b})}
  & (\ref{eq_secIII-2.002}) \& (\ref{eq_secIII-2.003-a})
  & (\ref{eq_secIII-3.002}) \& (\ref{eq_secIII-3.003-a}) \\ \cline{1-2} \cline{5-6}
  $\eta-\mu$ & $\eta-\mu$ &                     &
  & (\ref{eq_secIII-2.002}) \& (\ref{eq_secIII-2.003-b})
  & (\ref{eq_secIII-3.002}) \& (\ref{eq_secIII-3.003-b}) \\ \hline
  \end{tabular}
  }
  \end{table}

  \subsection{Average Bit Error Probability}

  Within the stochastic geometry framework, numerous studies have been conducted to evaluate the average BEP or symbol error probability (SEP). In \cite{Renzo2014, DiRenzo2014}, the authors used the equivalent-in-Distribution (EiD) approach which can treat the aggregate interference at the fine signal level and incorporate important communication attributes, such as modulation scheme and constellation size, into the modeling process. However, the mathematical framework of EiD approach has been developed under the assumption of a Rayleigh fading environment, making it difficult to utilize in generalized fading conditions.

  In this paper, to circumvent the dependency of the EiD approach upon Rayleigh fading, we use the alternative method proposed by \cite{Hamdi2007} (Theorem 1 and 2). The conditional bit error probability for instantaneous SINR $\gamma$ is evaluated as $\frac{\Gamma\left(b, a \gamma \right)}{2 \Gamma(b)}$ \cite{Wojnar1986},  where $a$ denotes the modulation type, $b$ represents the detection type\footnote{$a = \frac{1}{2}$ for orthogonal frequency shift keying (FSK), $a = 1$ for antipodal phase shift keying (PSK), $b = \frac{1}{2}$ for coherent detection, and $b = 1$ for non-coherent detection.}, and $\Gamma(a, x)$ is the upper incomplete gamma function. Then, the average BEP of an overlaid D2D network is given via the following theorem.

  \begin{thm}
    For an overlaid D2D network, the average BEP of a cellular UE ${P}_{e, c}$, a potential D2D UE in D2D mode $\hat{P}_{e, d}$,
   and a potential D2D UE ${P}_{e, d}$ are given by \cite{Choi2003}
    \begin{equation}
    \begin{split}
      \hat{P}_{e, d} &= \mathbb{E}\left[ \frac{\Gamma\left(b, \frac{a W}{I_d + N_0} \right)}{2 \Gamma(b)} \right], \quad
      {P}_{e, c} = \mathbb{E}\left[ \frac{\Gamma\left(b, \frac{a W}{I_c + N_0} \right)}{2 \Gamma(b)} \right], \\
      {P}_{e, d} &= \frac{R_c ~\mathrm{P}\left( L_d > \theta \right) {P}_{e, c}
      + \hat{R}_d ~\mathrm{P}\left( L_d \leq \theta \right)  \hat{P}_{e, d}}
      {R_c ~\mathrm{P}\left( L_d > \theta \right) + \hat{R}_d ~\mathrm{P}\left( L_d \leq \theta \right)},
      \end{split}
    \label{eq_secV-2.002}
    \end{equation}
  where $R_c$ and $\hat{R}_d$ are derived in (\ref{eq_secV-1.003}) and
  $\mathrm{P}\left( L_d \leq \theta \right)$ is evaluated in (\ref{eq_secII-3.001-ext1}).
  The average term in (\ref{eq_secV-2.002}) can be calculated as follows
  by substituting $g(x) = \frac{\Gamma\left(b, a x \right)}{2 \Gamma(b)}$ in Theorem 1 and 2
    \begin{equation}
    \begin{split}
      \mathbb{E}\left[ \frac{\Gamma\left(b, \frac{a W}{I + N_0} \right)}{2 \Gamma(b)} \right]
      &= \frac{1}{2} + \frac{\bar{w}}{\mu(1+\kappa) N_0}
      \sum_{n = 0}^{\infty} \frac{\left( \mu \kappa \right)^n}{n! ~ \mathrm{e}^{\mu \kappa}}
      \int_{0}^{\infty}
      g_{\mu+n}\left( \frac{\bar{w} x}{\mu(1+\kappa) N_0} \right)
      \mathcal{L}_{I}\left( \frac{x}{N_0}\right)
      \mathrm{e}^{-x} \mathrm{d}x\\
      &\simeq
  	\frac{1}{2} + \frac{\bar{w}}{\mu(1+\kappa) N_0} \sum_{m=1}^{M}
      \sum_{n = 0}^{\infty} \frac{c_m \left( \mu \kappa \right)^n}{n! ~ \mathrm{e}^{\mu \kappa}}
      g_{\mu+n}\left( \frac{\bar{w} x_m}{\mu(1+\kappa) N_0} \right)
      \mathcal{L}_{I}\left( \frac{x_m}{N_0}\right),
    \end{split}
    \label{eq_secV-2.003-a}
    \end{equation}
  for a $\kappa-\mu$ distributed signal envelope $\sqrt{G_i}= \sqrt{W}$, and
    \begin{equation}
    \begin{split}
      \mathbb{E}\left[ \frac{\Gamma\left(b, \frac{a W}{I + N_0} \right)}{2 \Gamma(b)} \right]
      &= \frac{1}{2} + \frac{\bar{w}}{2\mu h N_0}
      \sum_{n = 0}^{\infty} a_n
      \int_{0}^{\infty}
      g_{2 \mu+ 2 n}\left( \frac{\bar{w} x}{2\mu h N_0} \right)
      \mathcal{L}_{I}\left( \frac{x}{N_0}\right)
      \mathrm{e}^{-x} \mathrm{d}x\\
      &\simeq
  	\frac{1}{2} + \frac{\bar{w}}{2\mu h N_0}
      \sum_{m = 1}^{M}
      \sum_{n=0}^{\infty} a_n c_m g_{2\mu+2n}\left( \frac{\bar{w} x_m}{2\mu h N_0} \right)
      \mathcal{L}_{I}\left( \frac{x_m}{N_0} \right),
    \end{split}
    \label{eq_secV-2.003-b}
    \end{equation}
    for a $\eta-\mu$ distributed signal envelope,
  where $a_n$ is defined in (\ref{eq_secIV-3.004}), $c_m$ and $x_m$ are the $m$-th weight and abscissa of the $M$-th order Laguerre polynomial, respectively. The Laplace transforms of the interference $\mathcal{L}_{I}(s)$ are derived in (\ref{eq_secIII-2.002}) and (\ref{eq_secIII-3.002}) for D2D and cellular link§§§§§§§, respectively. The derivative terms $g_i(z)$ in (\ref{eq_secV-2.003-a}) and (\ref{eq_secV-2.003-b})
  are evaluated as follows
    \begin{equation}
    \begin{split}
      g_i(x) &= \frac{1}{2 \Gamma(b) \Gamma(i)} \frac{\mathrm{d}^i}{\mathrm{d}x^i}\left( x^{i-1} \Gamma\left(b, a x \right) \right)
      = \frac{a}{2 \Gamma(b) \Gamma(i)}
      G_{3, 2}^{1, 2}\left( \frac{1}{a x}
      \Bigg\vert { \mycom{2, 2-b, 1}{1+i, 1} } \right),
      \end{split}
    \label{eq_secV-2.004}
    \end{equation}
  where we have applied (\ref{eq_app_III-001}), (\ref{eq_app_III-005}), and (\ref{eq_app_III-004}) in the last equality.
  \end{thm}

  Similar to the rate analysis in (\ref{eq_secV-1.012})-(\ref{eq_secV-1.023}), the average BEP can be evaluated for the special cases of the $\kappa-\mu$ and $\eta-\mu$ fading models by substituting $g(x) = \frac{\Gamma\left(b, a x \right)}{2 \Gamma(b)}$ into Appendix VII and Theorem 4. If different fading distributions impact the intended and interfering links, then the appropriate measures should be used
  for evaluating the average BEP, as given in Table \ref{tab_II}.
  For example, if the received signal envelope of the intended link is $\eta-\mu$ distributed and
  that of the interference link is $\kappa-\mu$ distributed, then the average
  term $\mathbb{E}\left[ {\Gamma\left(b, a\gamma \right)}/{2 \Gamma(b)} \right]$ in
  (\ref{eq_secV-2.003-b}) can be evaluated by using the Laplace transform $\mathcal{L}_{I}(s)$
  in (\ref{eq_secIII-2.002}) and (\ref{eq_secIII-2.003-a})  for the D2D link, or
  (\ref{eq_secIII-3.002}) and (\ref{eq_secIII-3.003-a}) for the cellular link, respectively.

  \subsection{Numerical Results}

  In the following, we compare numerical results for different fading distributions. All of the simulations were carried out using the following parameters: BS node intensity $\lambda_b = \frac{1}{\pi 500^2}$, UE node intensity $\lambda = \frac{10}{\pi 500^2}$, ALOHA transmit probability $\varepsilon = 0.8$, D2D distance parameter $\xi = \frac{10}{\pi 500^2}$, path-loss exponent $\tau_c = \tau_d = 4$, spectrum partition factor $\beta = 0.2$, mode selection threshold $\theta = 100 m$, and  probability of potential D2D UEs $q = 0.2$. We compare the average rate versus the SNR for different fading environments when the intended UE is transmitting in D2D mode (\figref{fig_rate_special_d}) or cellular mode (\figref{fig_rate_special_c}). The average received SNR is determined as $\mathrm{SNR} = \frac{\bar{w}}{N_0}$ due to the power control and we note that the rate increases for a larger SNR, then decreases after a certain SNR threshold. This effect is due to the fact that every node is transmitting at a same SNR; at high SNR region, the interference also increases as the transmit power increases, degrading the overall network performance. We observe that D2D transmission achieves a higher rate than the cellular link for $\beta = 0.2$. Also, a dominant LOS component (large $\kappa$) or a large number of scattering clusters (large $\mu$) lead to higher average rates, whereas large $\eta$ or a small number of scattering clusters (small $\mu$) decreases the average rates.
 
  \begin{figure}[!t]
\centering
\subfigure[Average rate of D2D mode versus the SNR.]{%
  \label{fig_rate_special_d}%
  \includegraphics[width=\myimagewidth, height = 62 mm]{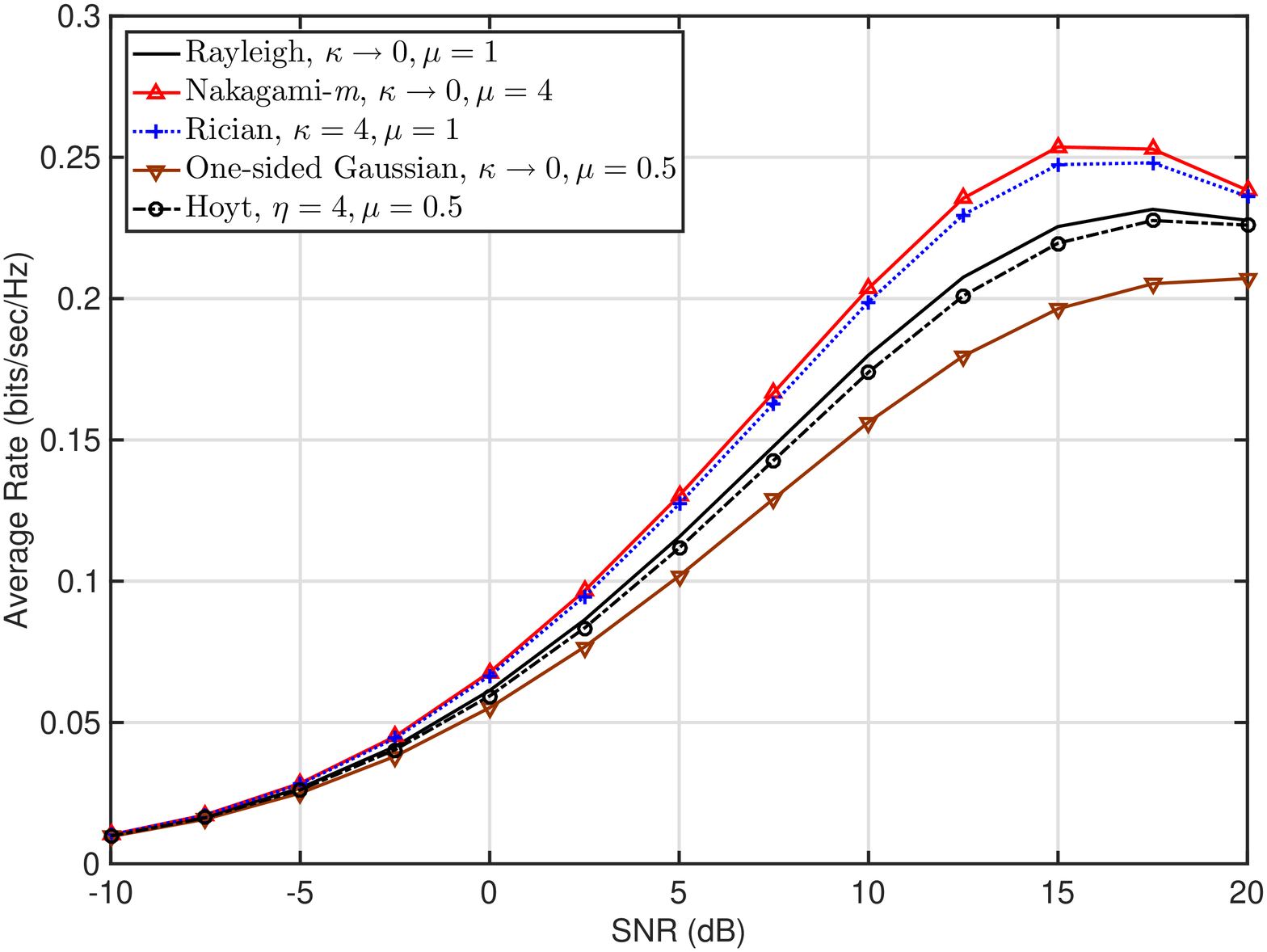}%
}
\subfigure[Average rate of cellular mode versus the SNR.]{%
  \label{fig_rate_special_c}%
  \includegraphics[width=\myimagewidth, height = 62 mm]{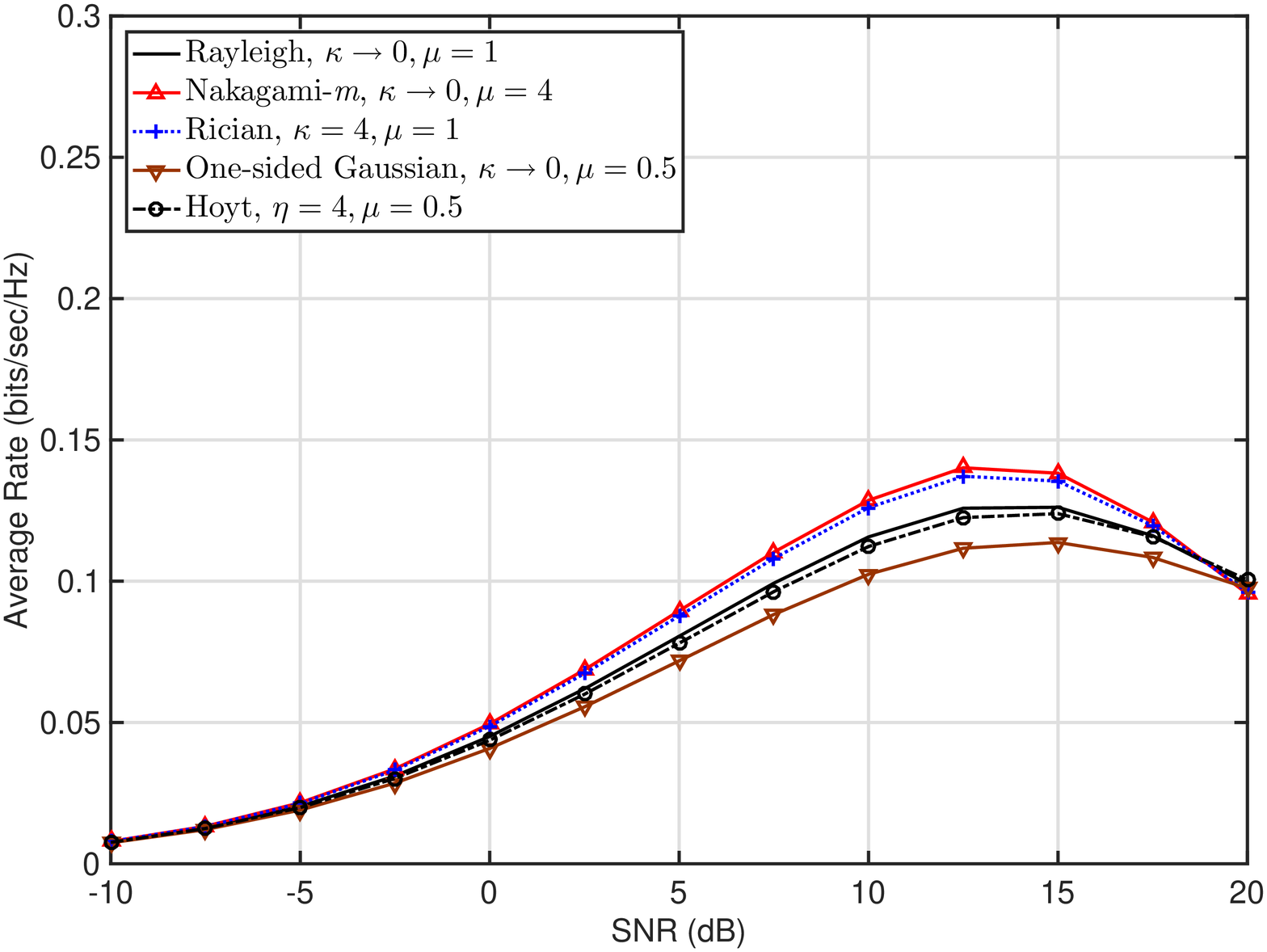}%
}
\subfigure[Complementary CDF of the SIR.]{%
  \label{fig1}%
  \includegraphics[width=\myimagewidth, height = 62 mm]{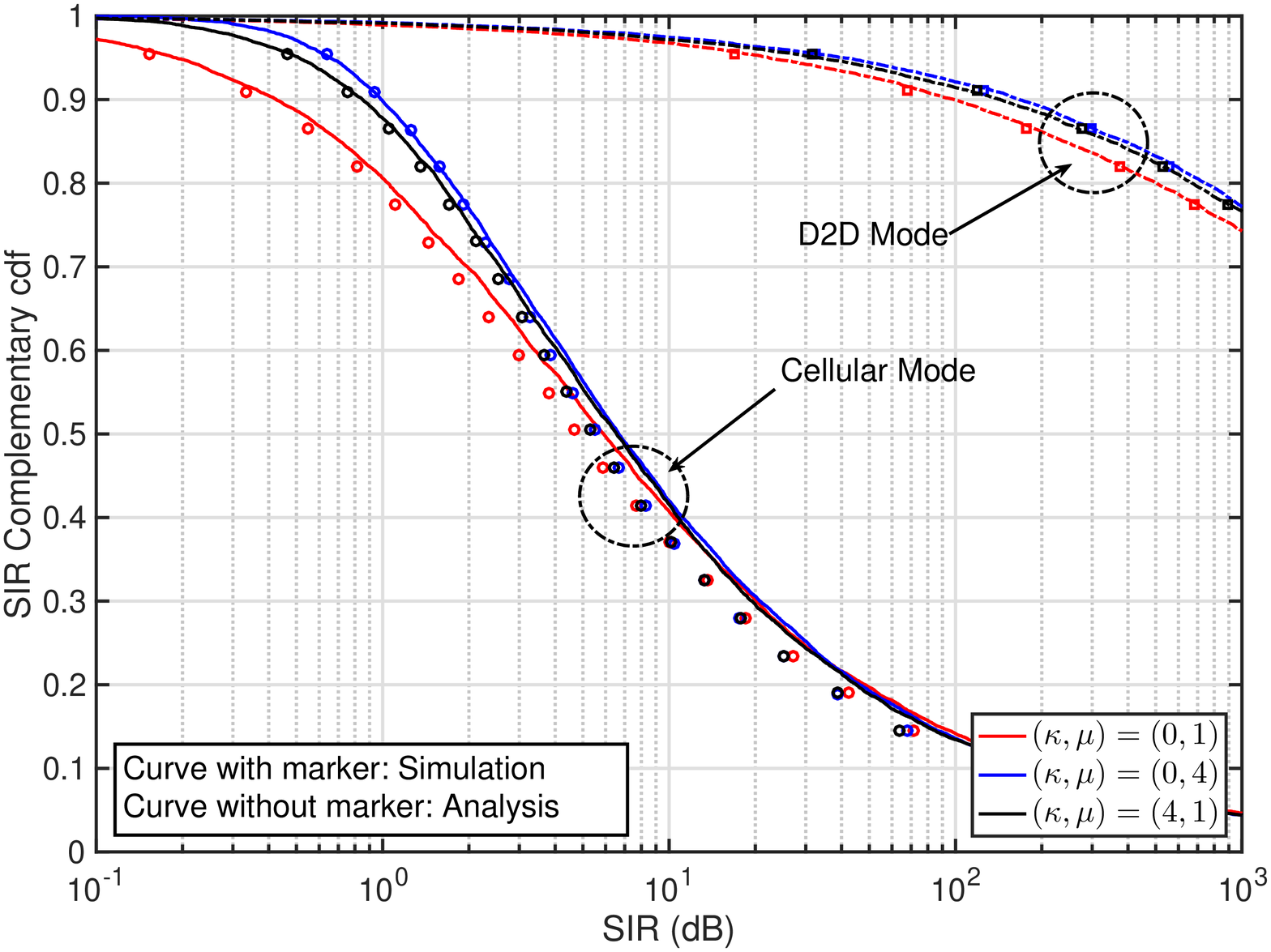}%
}
\subfigure[Average rate versus mode selection threshold $\theta$.]{%
  \label{fig5}%
  \includegraphics[width=\myimagewidth, height = 62 mm]{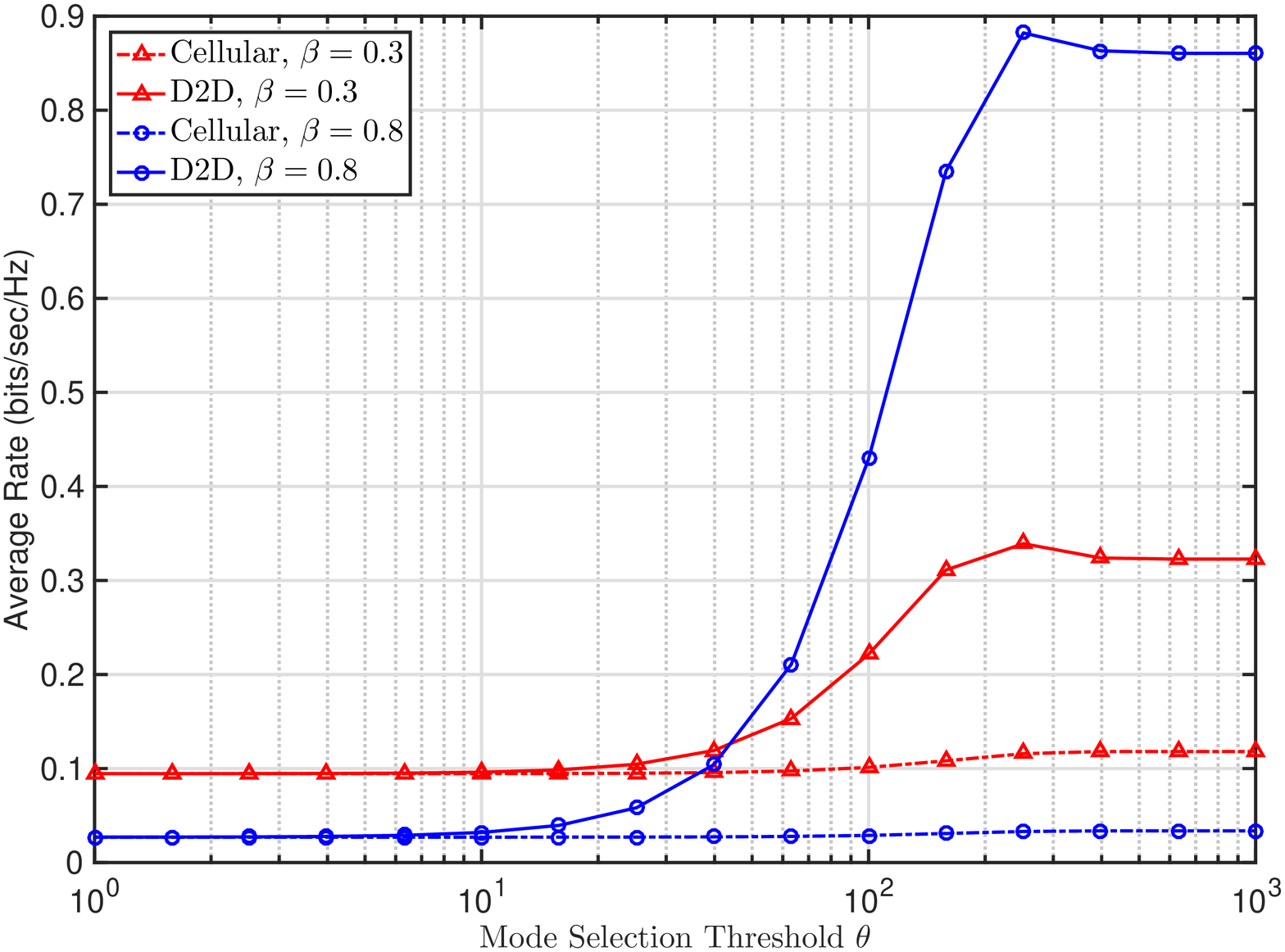}%
}
\subfigure[Average rate versus spectrum partition factor $\beta$.]{%
  \label{fig.rate_beta-1}%
  \includegraphics[width=\myimagewidth, height = 62 mm]{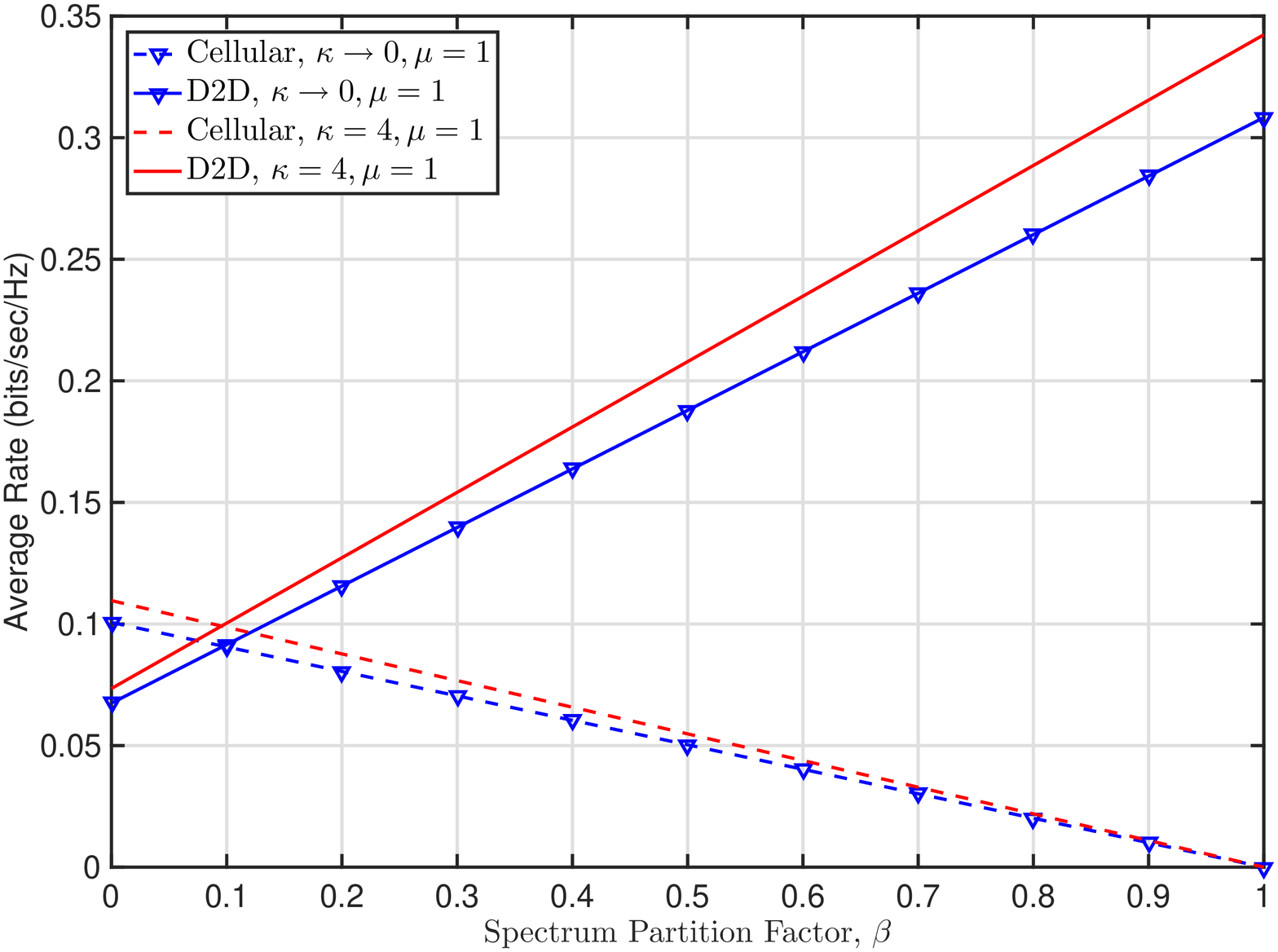}%
}
\subfigure[Average bit error probability versus the SNR.]{%
  \label{fig.BEP_left}%
  \includegraphics[width=\myimagewidth, height = 62 mm]{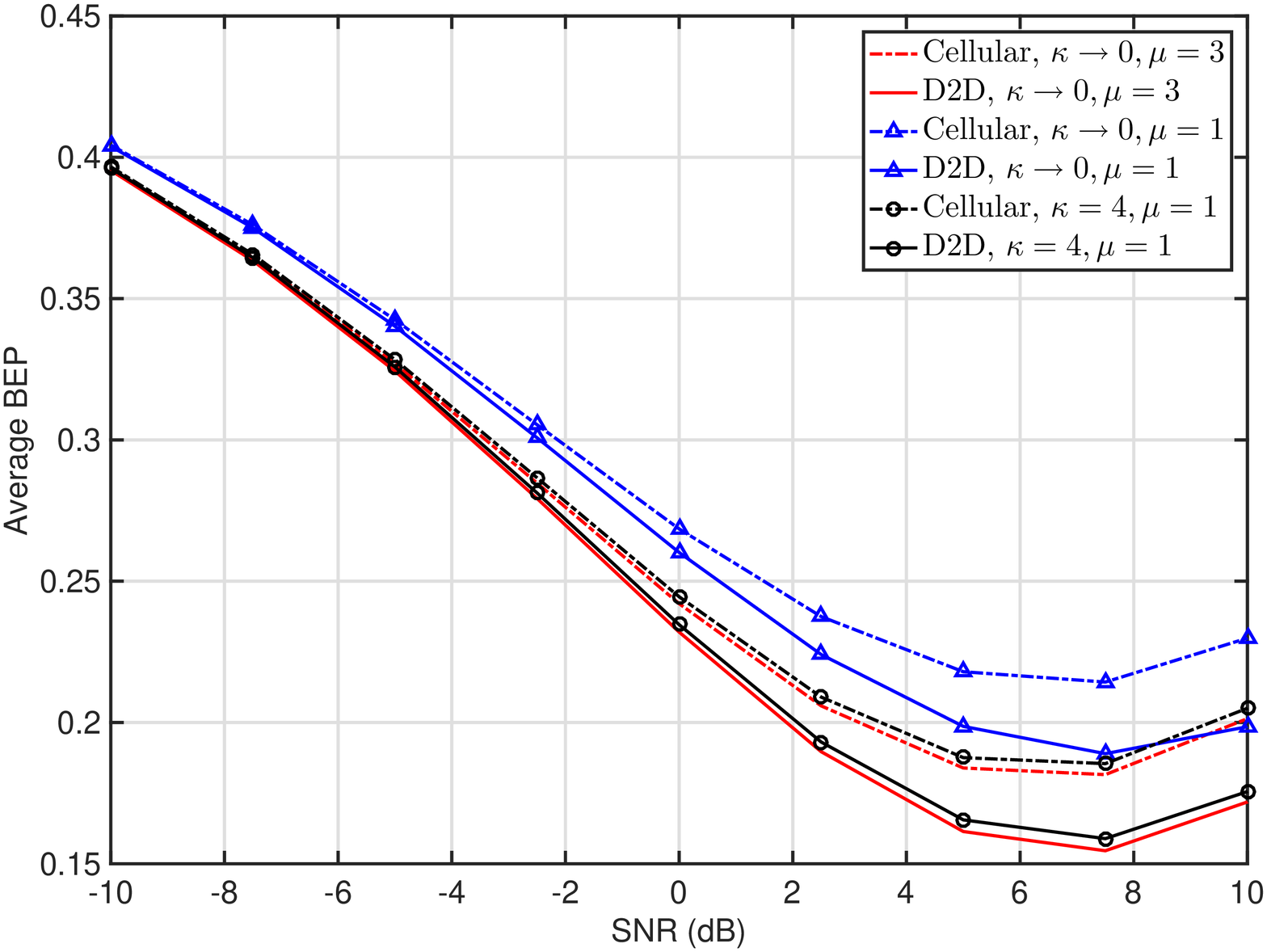}%
}
\caption{Simulation and Numerical Results.}%
\label{fig:Scatter}
\end{figure}

  \figref{fig1} plots the complementary CDF (CCDF) of the SIR for D2D and cellular links versus the SIR for several of the different fading environments contained within the new model proposed here. Monte Carlo simulation results are obtained by adopting a similar approach as \cite{Lin2014} where a hexagonal grid model is assumed for the cellular uplink and D2D nodes are distributed according to a PPP. The analytic closed form expression for the SIR CCDF can be derived either by using the proposed framework in (\ref{eq_secIV-3.001})-(\ref{eq_secIV-3.004}) with an indicator function $g(x) = \mathbb{I}\left( x \geq \gamma_{th}\right)$ or by using the series representation as follows
		\begin{equation}
		   \begin{split} 
			\mathbb{P}\left( \frac{W}{I} > x \right) &= 
			1 - \mathbb{E}_{I}\left[ 
			\sum_{n=0}^{\infty} c_n \left( x ~I \right)^{n + \mu}
			\right]
			= 1 - 	\sum_{n=0}^{\infty} c_n x^{n + \mu} \left. \frac{\partial^{n + \mu} \mathcal{L}_{I}(s)}{\partial s^{n + \mu}} \right|_{s = 0},
		   \end{split}
		   \label{eq_secV-3.001}
		 \end{equation} 
		where (\ref{eq_app_III-008}) is applied to the CDF of $\kappa-\mu$ fading\footnote{The CCDF of the SIR for $\eta-\mu$ fading can be derived similarly.} \cite{Yacoub2007a} in the first equality and Laplace transform property, \textit{i.e.}, $\mathbb{E}\left[ X^{n} \right] = \left.\frac{\partial^{n} \mathcal{L}_{X}(s)}{\partial s^{n}}\right|_{s = 0}$, is utilized in the second equality. The CCDF in (\ref{eq_secV-3.001}) can be evaluated by using $\mathcal{L}_{I}(s)$ in (\ref{eq_secIII-2.003-b}) for D2D link, or (\ref{eq_secIII-3.003-b}) for cellular link, respectively. The analytical results match the simulation results perfectly for D2D mode, whereas for cellular mode, there is a small discrepancy caused by the approximated uplink model. We observe that the performance gap remains tolerable for various fading environments indicating the accuracy of the approximation method in \cite{Lin2014}. On average, D2D links have a closer transmission range than cellular links, which leads to a higher SIR distribution for D2D links than cellular links.

    \figref{fig5} shows the effect of the mode selection threshold $\theta$ on the rate for Rayleigh fading, \textit{i.e.}, $\kappa = 0, \mu = 1$. As shown in the figure, increasing $\theta$ results in less potential D2D UEs choosing to operate in the cellular mode, leading to a higher average rate for the cellular link. We note that the average rate of a potential D2D UE increases for a larger $\theta$, then decreases after a certain optimal $\theta$ value due to the increased co-channel interference over the D2D link.

  As illustrated in \figref{fig_rate_special_d} and \ref{fig_rate_special_c}, D2D transmission achieves a higher rate than the cellular link for $\beta = 0.2$. We compare the average rate versus $\beta$ in \figref{fig.rate_beta-1} for a fixed SNR = $5$ dB when both $W$ and $I$ are $\kappa-\mu$ distributed. On average, D2D links have a closer transmission range than cellular links, hence if a minimum amount of spectrum is allocated to the D2D link, which is $\beta \geq 0.1$ as in \figref{fig.rate_beta-1}, D2D UEs achieve higher transmission rates than cellular UEs. On the other hand, if $\beta < 0.1$, D2D transmission does not have enough radio resources to achieve rate gains against the cellular link.

  \figref{fig.BEP_left} compares the average bit error probability of a D2D link to that of a cellular link. For the purposes of illustration, we assume that both $W$ and $I$ are $\kappa-\mu$ distributed with coherent BFSK ($a = b = 1/2$). A dominant LOS component (large $\kappa$) or a large number of scattering clusters (large $\mu$) improves the bit error probability. Note that the D2D link has worse error probability compared to that of the cellular link. This is attributed to the fact the cellular link uses orthogonal multiple access which eliminates interference within each cell coverage, whereas the D2D link uses random medium access, which improves the cellular link reliability at the cost of reduced rate.

  \section{Conclusion}

  In this paper, we have considered a D2D network overlaid on an uplink cellular network, where the spatial locations of the mobile UEs as well as the BSs are modeled as PPP. In particular, we have introduced a new stochastic geometric approach for evaluating the D2D network performance under the assumption of generalized fading conditions described by the $\kappa-\mu$ and $\eta-\mu$ fading models. Using these methods, we evaluated the average rate and average bit error probability of the overlaid D2D network. Specifically, we observed that the D2D link provides higher rates than those of the cellular link when the spectrum partition factor was appropriately chosen. Under these circumstances, setting a large mode selection threshold will encourage more UEs to use the D2D mode, which increases the average rate at the cost of a higher level of interference and degraded bit error probability. However, for smaller values of the spectrum partition factor, the D2D link has smaller rates than those of the cellular link. In terms of the fading parameters, a dominant LOS component (large $\kappa$) or a large number of scattering clusters (large $\mu$) improve the network performance, \textit{i.e.}, a higher rate and lower BEP are achieved, whereas large $\eta$ or a small number of scattering clusters (small $\mu$) deteriorate the performance. Finally, we also provided numerical results to demonstrate the performance gains of overlaid D2D networks compared to traditional cellular networks, where the latter corresponds to $\beta = 0$ case.




  \section*{Appendix I}

  In this appendix, we summarize operational equalities of the special functions, which are used in this paper.
  First, the Meijer G-function is denoted by \cite{Gradshteyn1994}
          \begin{equation}
          \begin{split}
          &\MeijerG[a, b]{n}{p}{m}{q}{z} =
          G_{p,q}^{m,n}\left( x \Bigg\vert { \mycom{a_p}{b_q} } \right),
          \end{split}
          \label{eq_app_III-000}
          \end{equation}
  and has the following operational identities on multiplication, inverse, integration, and differentiation
          \begin{equation}
          \begin{split}
          \text{Multiplication:} \quad
          x^k G_{p,q}^{m,n}\left( x \Bigg\vert { \mycom{a_p}{b_q} } \right)
          = G_{p,q}^{m,n}\left( x \Bigg\vert { \mycom{a_p + k}{b_q + k} } \right),
          \end{split}
          \label{eq_app_III-001}
          \end{equation}
          \begin{equation}
          \begin{split}
          \text{Inversion:} \quad
          G_{p,q}^{m,n}\left( \frac{1}{x} \Bigg\vert { \mycom{a_p}{b_q} } \right)
          = G_{q,p}^{n,m}\left( x \Bigg\vert { \mycom{1-b_q}{1-a_p} } \right),
          \end{split}
          \label{eq_app_III-002}
          \end{equation}
          \begin{equation}
          \begin{split}
          \text{Integration:} \quad
          \int G_{p,q}^{m,n}\left( x \Bigg\vert { \mycom{a_p}{b_q} } \right) \mathrm{d}x
          = G_{p+1, q+1}^{m, n+1}\left( x \Bigg\vert { \mycom{1, (a_p+1)}{(b_q+1, 0)} } \right),
          \end{split}
          \label{eq_app_III-003}
          \end{equation}
          \begin{equation}
          \begin{split}
          \text{Differentiation:} \quad
          \frac{\partial^{\nu}}{\partial x^{\nu}}
          G_{p,q}^{m,n}\left( \frac{1}{x} \Bigg\vert { \mycom{a_p}{b_q} } \right)
          =
          \frac{1}{x^{\nu}}
          G_{q,p}^{n,m}\left( x \Bigg\vert {
          \mycom{0, a_1, a_2, \ldots, a_n, a_{n+1}, \ldots, a_p}{b_1, b_2, \ldots, b_m, \nu, b_{m+1}, \ldots, b_q} } \right).
          \end{split}
          \label{eq_app_III-004}
          \end{equation}
  The Meijer G-function can represent elementary functions or be simplified as follows
          \begin{equation}
          \begin{split}
          \Gamma(a, b) = G_{1, 2}^{2, 0}\left( b \Bigg\vert { \mycom{1}{0, a} } \right),
          \quad
          G_{1, 1}^{1, 1}\left( x \Bigg\vert { \mycom{a}{b} } \right) =
          \Gamma(1-a+b) x^b \left( x+1 \right)^{a-b-1}.
          \end{split}
          \label{eq_app_III-005}
          \end{equation}
  The following properties of Gamma function hold for non-negative real constants $x$ and $y$
          \begin{equation}
          \begin{split}
          \Gamma(x) \Gamma\left(x + \frac{1}{2}\right) = 2^{1-2x} \sqrt{\pi} \Gamma(2x)&,
          ~
          \binom{x}{y} = \frac{\Gamma(x+1)}{\Gamma(y+1)\Gamma(x-y+1)}, \\
         \Gamma(1+x)\Gamma(1-x) = \frac{1}{\mathrm{sinc}\left( x \right)}&,
         ~ \Gamma\left( \frac{1}{2}\right) = \sqrt{\pi}.
          \end{split}
          \label{eq_app_III-007}
          \end{equation}
The generalized Marcum Q-function can be represented in power series using Laguerre polynomials \cite{Andras2010}
          \begin{equation}
          \begin{split}
			Q_{\mu}\left[ \sqrt{2 a}, 
			\sqrt{2 b ~x} \right] &= 
			1 - \sum_{n=0}^{\infty}c_n x^{n+\mu},
          \end{split}
          \label{eq_app_III-008}
          \end{equation}
for arbitrary $a > 0$, $b > 0$, and $x \geq 0$ where the coefficient $c_n$ denotes the following expression
		\begin{equation}
          \begin{split}
			c_n &\triangleq (-1)^{n}\mathrm{e}^{-a}
			\frac{L_n^{\mu-1}(a) b^{n+\mu}}{\Gamma(n+\mu+1)},
          \end{split}
          \label{eq_app_III-008-2}
          \end{equation}
and $L_n^{\mu-1}$ is the generalized Laguerre polynomial of degree $n$ and order $\mu-1$.

  By using the property of the binomial coefficient, $\binom{n + r}{r} = \binom{n+r}{n} = (-1)^n \binom{-r-1}{n}$,
  and the Taylor series, $(1+x)^{\mu} = \sum_{n=0}^{\infty}\binom{\mu}{n} x^n$,
  the following summation can be simplified as
          \begin{equation}
          \begin{split}
         \sum_{n=0}^{\infty} \binom{n+\mu-1}{n} x^{2n} = \sum_{n=0}^{\infty} (-1)^n \binom{-\mu}{n}  x^{2n}
         = \left( 1 - x^2 \right)^{-\mu},
           \end{split}
          \label{eq_app_III-010}
          \end{equation}
  which holds for arbitrary real constant $-1 < x < 1$ and a complex constant $\mu$.

  \section*{Appendix II}

  In this appendix, we summarize the moments of UE's transmit power, which is proved in \cite[Lemma 1]{Lin2014}, to provide the reader with a better understanding of our derivation process. Note that a potential D2D UE can choose either a D2D mode or a cellular mode based on the mode selection scheme. For notational clarity, we denote the transmit power of the potential D2D UE in D2D mode as $\hat{P}_d$ and that of the potential D2D UE as $P_d$.
  \begin{lemm}
  The $n$-th moments of the transmit power of a cellular UE, a potential D2D UE in D2D mode, and a potential D2D UE are given by
  	  \begin{equation}
  	  \begin{split}
  	    \mathbb{E}\left[ {P_c}^n \right] &= \frac{R^{\tau_c n}}{1 + \frac{\tau_c n}{2}}, \quad
  	    \mathbb{E}\left[ {\hat{P}_d}^n \right] = \frac{ \left( \xi \pi \right)^{-\frac{\tau_d n}{2}} }{\mathrm{P}\left( L_d \leq \theta \right)}
  	    \gamma\left( 1 + \frac{\tau_d n}{2}, \xi \pi \theta^2 \right),\\
  	    \mathbb{E}\left[ {P_d}^n \right] &= \left( \xi \pi \right)^{-\frac{\tau_d n}{2}} \gamma\left( 1 + \frac{\tau_d n}{2}, \xi \pi \mu^2 \right)
  	    + \mathrm{P}\left( L_d > \theta \right) \frac{R^{\tau_c n}}{1 + \frac{\tau_c n}{2}},
  	  \end{split}
  	  \label{eq_secIII-1.003}
  	  \end{equation}
  	  where $n > 0$ has positive real value, $\xi$ is the D2D distance parameter, $R = \sqrt{\frac{1}{\pi \lambda_b}}$ is the cellular range, and $\tau_c$ ($\tau_d$) is the path-loss exponent over cellular (D2D) link. We note that $P_d$ includes the potential D2D UEs in D2D mode as well as the ones in cellular mode.
  \end{lemm}

  \section*{Appendix III}

  In this appendix, we provide a proof for Lemma $2$.
  For the $\kappa-\mu$ distribution, we obtain (\ref{eq_secII-2.010}) by substituting the second relation of (\ref{eq_app_III-005}) into (\ref{eq_secII-2.003}) as follows
    \begin{equation}
    \begin{split}
      \mathcal{L}_{G}(s) &=
      \frac{1}{\mathrm{e}^{\mu \kappa}} \sum_{n = 0}^{\infty} \frac{\left( \mu \kappa \right)^n}{n!}
      \left( \frac{1}{1 + \frac{s \bar{w}}{\mu (1 + \kappa)}} \right)^{n + \mu}
      = \frac{1}{\left( 1 + \frac{s \bar{w}}{\mu (1 + \kappa)} \right)^{\mu}}
      \exp\left( -\frac{\mu \kappa}{1 + \frac{\mu (1 + \kappa)}{s \bar{w}}}\right),
    \end{split}
    \label{eq_app_II-ext-001}
    \end{equation}
    where we used the Taylor expansion of an exponential function, \textit{i.e.}, $\mathrm{e^{x}} = \sum_{n=0}^{\infty}\frac{x^n}{n!}$,
    in the last equality.
  For the $\eta-\mu$ distribution, we simplify (\ref{eq_secII-2.007}) as follows
    \begin{equation}
    \begin{split}
      \mathcal{L}_{G}(s)
       &= \frac{1}{h^{\mu}}
       \frac{1}{\left( 1 + \frac{s \bar{w}}{2 \mu h} \right)^{2 \mu}}
       \sum_{n = 0}^{\infty}
       \binom{n + \mu - 1}{n}
       \left( \frac{H}{h}\right)^{2n}
       {\left( 1 + \frac{s \bar{w}}{2 \mu h} \right)^{-2 n}}\\
       &= \frac{1}{h^{\mu}}
       \frac{1}{\left( 1 + \frac{s \bar{w}}{2 \mu h} \right)^{2 \mu}}
       \left( 1 - \left( \frac{H/h}{1 + \frac{s \bar{w}}{2 \mu h}}\right)^2 \right)^{-\mu},
    \end{split}
    \label{eq_app_II-ext-002}
    \end{equation}
    where we applied (\ref{eq_app_III-005}) and (\ref{eq_app_III-007}) in the first equality. Since $H$ and $h$ are related by the following expression,
  \begin{equation}
    \begin{split}
    \frac{H}{h} &= \frac{\eta^{-1}-\eta}{2 + \eta^{-1} + \eta} = \frac{1-\eta}{1+\eta} = \frac{2}{1+\eta} - 1, \quad \eta > 0,
    \end{split}
    \label{eq_app_II-ext-003}
    \end{equation}
    $H/h$ is a decreasing function of $\eta$ with a magnitude $-1 < \frac{H}{h} < 1$. Hence,
    $-1 < \frac{H/h}{1 + \frac{s \bar{w}}{2 \mu h}} < 1$ holds for a non-negative valued $\frac{s \bar{w}}{2 \mu h}$ and
    (\ref{eq_app_III-010}) can be applied to the last equality,
    which is equivalent to (\ref{eq_secII-2.010}). This completes the proof.

  \section*{Appendix IV}

  In this appendix, we provide a proof for Lemma $4$.
  The Laplace transform of the interference at the D2D receiver is evaluated as follows
    \begin{equation}
    \begin{split}
      \mathcal{L}_{I_d}(s) &= \mathbb{E}\left[ \mathrm{e}^{-s I_d} \right]
      = \mathbb{E}\left[ \exp\left(-s \sum_{X_j \in \varepsilon \Phi_d \backslash \{ 0 \}}
      \hat{P}_{d, j} G_j {||X_j||}^{-\tau_d} \right) \right]\\
      &= \mathbb{E}\left[ \prod_{X_j \in \varepsilon \Phi_d}
          \mathbb{E}\left[
            \exp\left(-s \hat{P}_{d, j} G_j {||X_j||}^{-\tau_d} \right)
          \right]
        \right]\\
      &= \exp\left( -2 \pi \varepsilon \lambda_d
        \int_{0}^{\infty} \left( 1 - \mathbb{E}\left[
          \exp\left(-s \hat{P}_{d} G {r}^{-\tau_d} \right)
        \right] \right) r\mathrm{d}r
      \right),
    \end{split}
    \label{eq_app_IV-001}
    \end{equation}
  where we have substituted (\ref{eq_secIII-2.001}) into the second equality, used Slivnyak's theorem in the third equality \cite{Haenggi2013},
  applied the probability generating functional (PGFL) of PPP in the last equality \cite{Haenggi2013}.

  By applying a change of variable, \textit{i.e.}, $s \hat{P}_{d} G {r}^{-\tau_d} = t$,
  and integration by parts in the last equality of (\ref{eq_app_IV-001}), the Laplace transform $\mathcal{L}_{I_d}(s)$ simplifies as follows
    \begin{equation}
    \begin{split}
      \mathcal{L}_{I_d}(s) &= \exp\left( - \pi \varepsilon \lambda_d
      \mathbb{E}\left[
      \left( s \hat{P}_{d} G \right)^{\delta_d}
      \int_{0}^{\infty} \delta_d t^{-\delta_d - 1} \left( 1 - \mathrm{e}^{-t} \right) \mathrm{d}t\right] \right) \\
      &=
      \exp\left( - \pi \varepsilon \lambda_d
      s^{\delta_d} \Gamma\left( 1 - \delta_d \right)
      \mathbb{E}\left[ {\hat{P}_{d}}^{\delta_d} \right]
      \mathbb{E}\left[ G^{\delta_d} \right] \right),
    \end{split}
    \label{eq_app_IV-003}
    \end{equation}
  where $\delta_d = \frac{2}{\tau_d}$ is used in the first equality, and
  the Gamma function $\Gamma(t) = \int_{0}^{\infty}x^{t-1} \mathrm{e}^{-x} \mathrm{d}x$ is applied in the second equality.
  Then, by substituting $\lambda_d = q \mathrm{P}(L_d \leq \theta) \lambda$, (\ref{eq_secII-2.003}), (\ref{eq_secIII-1.003}) into (\ref{eq_app_IV-003}),
  the Laplace transform of $I_d$ for the $\kappa-\mu$ distribution is evaluated as
    \begin{equation}
    \begin{split}
      \mathcal{L}_{I_d}(s)
      &=
      \exp\left( -
      \frac{q \varepsilon \lambda s^{\delta_d}}{\xi \mathrm{e}^{\mu \kappa}}
      \left( \frac{\bar{w}}{[(1+\kappa)\mu]} \right)^{\delta_d}
      \frac{
  	\gamma\left( 2, \xi \pi \theta^2 \right)
      \Hypergeometric{1}{1}{\mu+\delta_d}{\mu}{\mu \kappa}}{\mathrm{sinc}\left(\delta_d\right)}
      \binom{\mu+\delta_d-1}{\delta_d} \right),
    \end{split}
    \label{eq_app_IV-004}
    \end{equation}
  where we used (\ref{eq_app_III-007}). Hence, we obtain (\ref{eq_secIII-2.003-a}) for the $\kappa-\mu$ distribution. The Laplace transform of $I_d$ for the $\eta-\mu$ distribution can be proved in a similar manner. This completes the proof.

  \section*{Appendix V}

  In this appendix, we provide a proof for Lemma $5$.
  The Laplace transform of the interference at the cellular BS is evaluated by (\ref{eq_secIII-3.001}) as follows
    \begin{equation}
    \begin{split}
      \mathcal{L}_{I_c}(s) &= \mathbb{E}\left[ \exp\left(-s \sum_{X_j \in \Phi_{c, a} \cap \mathcal{A}^c}
      {P}_{c, j} G_j {||X_j||}^{-\tau_c} \right) \right]\\
      &= \mathbb{E}\left[ \prod_{X_j \in \Phi_{c, a}}
      \exp\left(-s {P}_{c, j} G_j {||X_j||}^{-\tau_c} ~ \mathbb{I}\left( ||X_j|| \geq R \right)  \right) \right]\\
      &= \exp\left( -2 \pi \lambda_b
        \int_{R}^{\infty} \left( 1 - \mathbb{E}\left[
          \exp\left(-s {P}_{c} G {r}^{-\tau_c} \right)
        \right] \right) r\mathrm{d}r
      \right),
    \end{split}
    \label{eq_app_V-001}
    \end{equation}
  where $\mathbb{I}(x)$ is an indicator function, \textit{i.e.}, $\mathbb{I}(x) = 1$ for $x > 0$ and $0$ otherwise,
  $\mathcal{A}^c$ indicates the region outside the cellular coverage, and
  the PGFL of PPP \cite{Haenggi2013} is used in the last equality.

  The average term in the last equality of (\ref{eq_app_V-001}) is denoted by $\varphi(r) \triangleq
  \mathbb{E}\left[ \mathrm{e}^{-s {P}_{c} {r}^{-\tau_c} G} \right]$ where the expectation is derived over the ensemble $P_c$ and $G$.
  For the $\kappa-\mu$ distribution, $\varphi(r)$ is evaluated as follows
    \begin{equation}
    \begin{split}
        \varphi(r) &=
        \mathbb{E}_{P_c, G}\left[ \mathrm{e}^{-s {P}_{c} {r}^{-\tau_c} G} \right]
        =
        \mathbb{E}_{P_c}\left[
          \mathcal{L}_{G}\left( s P_c {r}^{-\tau_c} \right)
        \right]\\
        &=
        \mathbb{E}_{L_c}\left[
          \mathcal{L}_{G}\left( s {L_c}^{\tau_c} {r}^{-\tau_c} \right)
        \right] =
        \int_{0}^{R} \mathcal{L}_{G}\left( s x^{\tau_c} {r}^{-\tau_c} \right) f_{L_c}(x) \mathrm{d}x,
      \end{split}
      \label{eq_app_V-002}
      \end{equation}
  where averaging $\mathrm{e}^{-n G}$ over the intended channel $G$ achieves the Laplace transform of $G$ in the second equality
  and the channel inversion based power control, \textit{i.e.}, $P_c = L_c^{\tau_c}$, is applied in the third equality.
  Then, by substituting (\ref{eq_secII-1.003}) and (\ref{eq_secII-2.007}) into (\ref{eq_app_V-002}),
  $\varphi(r)$ is calculated as below
    \begin{equation}
    \begin{split}
        &\varphi(r) =
        \frac{2\pi \lambda_b}{\mathrm{e}^{\mu \kappa}}
        \sum_{n = 0}^{\infty} \frac{\left( \mu \kappa \right)^n}{n! \Gamma(n+\mu)}
        \int_{0}^{R} x G_{1,1}^{1,1}\left( \frac{\mu (1+\kappa)}{s x^{\tau_c} {r}^{-\tau_c} \bar{w}}
        \Bigg\vert { \mycom{1}{n+\mu} } \right) \mathrm{d}x\\
        &= \frac{2\pi \lambda_b}{\mathrm{e}^{\mu \kappa}}
        \sum_{n = 0}^{\infty} \frac{\left( \mu \kappa \right)^n}{n! \Gamma(n+\mu)}
        \int_{0}^{R} x G_{1,1}^{1,1}\left( \frac{s x^{\tau_c} {r}^{-\tau_c} \bar{w}}{\mu (1+\kappa)}
        \Bigg\vert { \mycom{1-n-\mu}{0} } \right) \mathrm{d}x\\
        &= \frac{\pi \lambda_b r^2 \delta_c}{\mathrm{e}^{\mu \kappa}}
        \left( \frac{\mu (1 + \kappa)}{s \bar{w}} \right)^{\delta_c}
        \sum_{n = 0}^{\infty} \frac{\left( \mu \kappa \right)^n}{n! \Gamma(n+\mu)}
        \int_{0}^{\frac{s \bar{w} R^{\tau_c}}{\mu (1+\kappa) r^{\tau_c}}}
        t^{\delta_c-1} G_{1,1}^{1,1}\left( t
        \Bigg\vert { \mycom{1-n-\mu}{0} } \right) \mathrm{d}t\\
        &= \frac{\pi \lambda_b r^2 \delta_c}{\mathrm{e}^{\mu \kappa}}
        \left( \frac{\mu (1 + \kappa)}{s \bar{w}} \right)^{\delta_c}
        \sum_{n = 0}^{\infty} \frac{\left( \mu \kappa \right)^n}{n! \Gamma(n+\mu)}
        G_{2,2}^{1,2}\left( \frac{s \bar{w} R^{\tau_c}}{\mu (1+\kappa) r^{\tau_c}}
        \Bigg\vert { \mycom{1, 1+\delta_c-n-\mu}{\delta_c, 0} } \right)\\
        &= \frac{\pi \lambda_b r^2 \delta_c}{\mathrm{e}^{\mu \kappa}}
        \left( \frac{\mu (1 + \kappa)}{s \bar{w}} \right)^{\delta_c}
        \sum_{n = 0}^{\infty} \frac{\left( \mu \kappa \right)^n}{n! \Gamma(n+\mu)}
        G_{2,2}^{2,1}\left( \frac{\mu (1+\kappa) r^{\tau_c}}{s \bar{w} R^{\tau_c}}
        \Bigg\vert { \mycom{1 - \delta_c, 1}{0, n+\mu-\delta_c} } \right),
    \end{split}
    \label{eq_app_V-003}
    \end{equation}
  where we have applied (\ref{eq_app_III-002}) in the second and last equality, used a change of variable, \textit{i.e.},
  $\delta_c = \frac{2}{\tau_c}$ and $\frac{s x^{\tau_c} {r}^{-\tau_c} \bar{w}}{\mu (1+\kappa)} = t$, in the third equality,
  and applied (\ref{eq_app_III-003}) in the fourth equality.
  Hence, by substituting $R = {\left( \pi \lambda_b \right)^{-\frac{1}{2}}}$ into (\ref{eq_app_V-003}),
  we obtain (\ref{eq_secIII-3.003-a}) for the $\kappa-\mu$ distribution. The corresponding Laplace transform of $I_c$
  for the $\eta-\mu$ distribution can be obtained in a similar manner. This completes the proof.

  \section*{Appendix VI}

  In this appendix, we provide a proof for Theorem $3$.
  Due to the power control scheme adopted here, the received signal power $W$ is equal to the intended channel, \textit{i.e.}, $W = G_i$, where the PDF of $G_i$ is given by
  (\ref{eq_secII-2.003}). Then, the average of an arbitrary function of the SINR $\gamma = \frac{W}{I+N_0}$ is written as follows
    \begin{equation}
    \begin{split}
      \mathbb{E}\left[ \left. g\left( \frac{W}{I+N_0}\right) \right\vert I \right] &=
      \int_{0}^{\infty} g\left( \frac{x}{I+N_0}\right) f_{W}(x) \mathrm{d}x = \int_{0}^{\infty} g\left( \frac{x}{I+N_0}\right) f_{G_i}(x) \mathrm{d}x.
    \end{split}
    \label{eq_app_VIII-001}
    \end{equation}
  We express the Bessel function in (\ref{eq_secII-2.003}) by the series form and denote the following term as $a_n$
    \begin{equation}
    \begin{split}
    \mathrm{I}_{\nu}(x) = \sum_{n = 0}^{\infty} \frac{1}{n! ~ \Gamma(n+\nu+1)} \left( \frac{x}{2}\right)^{2n + \nu}, \quad
    a_n = \frac{(\mu \kappa)^n}{n! \mathrm{e}^{\mu \kappa}}.
    \end{split}
    \label{eq_app_VIII-002}
    \end{equation}
  Then, (\ref{eq_app_VIII-001}) is given by
    \begin{equation}
    \begin{split}
      \mathbb{E}\left[ \left. g\left( \frac{W}{I+N_0}\right) \right\vert I \right] &=
      \sum_{n = 0}^{\infty} a_n
      \int_{0}^{\infty} \frac{x^{\mu+n-1}}{\Gamma(\mu+n)} g\left( \frac{x}{I+N_0}\right)
      \left( \frac{\mu(1+\kappa)}{\bar{w}} \right)^{\mu + n}
      \mathrm{e}^{-\frac{\mu(1+\kappa)x}{\bar{w}}} \mathrm{d}x\\
      &=     \sum_{n = 0}^{\infty} a_n
      \int_{0}^{\infty} \frac{z^{\mu+n-1}}{\Gamma(\mu+n)} g\left( z \right)
      b^{\mu + n} \mathrm{e}^{-b z} \mathrm{d}x,
    \end{split}
    \label{eq_app_VIII-003}
    \end{equation}
  where we used a change of variable, \textit{i.e.}, $\frac{x}{I+N_0} = z$ and $b = \frac{\mu(1+\kappa)\left( I + N_0\right)}{\bar{w}}$,
  in the second equality.

  The integral in (\ref{eq_app_VIII-003}) can be evaluated as follows
    \begin{equation}
    \begin{split}
    \int_{0}^{\infty} \underbrace{\frac{z^{\mu+n-1}}{\Gamma(\mu+n)} g\left( z \right)}_{u} ~
      \underbrace{
      \vphantom{\frac{z^{\mu+n-1}}{\Gamma(\mu+n)} g\left( z \right)}
      b^{\mu + n} \mathrm{e}^{-b z}}_{v'} \mathrm{d}x
      &=
      \left. -\sum_{i = 0}^{\mu+n-1} g_i(z) b^{\mu+n-i-1} \mathrm{e}^{-bz} \right\vert_{0}^{\infty} + \int_{0}^{\infty}
      g_{\mu+n}(z)\mathrm{e}^{-bz} \mathrm{d}z,
    \end{split}
    \label{eq_app_VIII-004}
    \end{equation}
  where we applied integration by parts $\mu+n$ times, defined $g_i(z)$ in (\ref{eq_secIV-3.002}), and
    \begin{equation}
      g_i(0) =
      \begin{dcases}
      0, & \text{for } i < \mu+n-1\\
      g(0), & \text{for } i = \mu+n-1
      \end{dcases}.
    \label{eq_app_VIII-005}
    \end{equation}
  Then, the average of an arbitrary function of the SINR $\gamma = \frac{W}{I+N_0}$ is given by
    \begin{equation}
    \begin{split}
      \mathbb{E}\left[ g\left( \frac{W}{I+N_0}\right) \right] &=
      \mathbb{E}\left[ \mathbb{E}\left[ \left. g\left( \frac{W}{I+N_0}\right) \right\vert I \right]\right]\\
      &= g(0) + \sum_{n=0}^{\infty} \int_{0}^{\infty}
      a_n g_{\mu+n}\left( z \right)  \mathrm{e}^{-\frac{\mu(1+\kappa) N_0}{\bar{w}}z}
      \mathbb{E}_I\left[ \mathrm{e}^{-\frac{\mu(1+\kappa) I}{\bar{w}}z} \right] \mathrm{d}z\\
      &= g(0) + \frac{\bar{w}}{\mu(1+\kappa) N_0} \sum_{n=0}^{\infty} \int_{0}^{\infty}
      a_n g_{\mu+n}\left( \frac{\bar{w}}{\mu(1+\kappa) N_0}x \right)  \mathrm{e}^{-x}
      \mathcal{L}_I\left( \frac{x}{N_0} \right) \mathrm{d}x,
    \end{split}
    \label{eq_app_VIII-006}
    \end{equation}
  where we used $\sum_{n=0}^{\infty}a_n = 1$ in the first equality and a change of variable $\frac{\mu(1+\kappa) N_0 z}{\bar{w}} = x$
  in the last.

  An analytic function $g(x)$ of $x$ can be computed by using the Laguerre polynomial as follows
     \begin{equation}
    \begin{split}
      \int_{0}^{\infty} \mathrm{e}^{-x} g(x) \mathrm{d}x = \sum_{m=1}^{M} c_m g\left( x_m \right) + R_M,
    \end{split}
    \label{eq_app_VI-003}
    \end{equation}
  where $x_m$ and $c_m$ are the $m$-th abscissa and weight of the $M$-th order Laguerre polynomial, respectively.
  The remainder $R_M$ rapidly converges to zero \cite{Gradshteyn1994}. Hence,
  the series expression follows by using the Laguerre polynomial in (\ref{eq_app_VI-003}) and this completes the proof.

  \section*{Appendix VII}

  In this appendix, we provide a proof for (\ref{eq_secV-1.012}) and (\ref{eq_secV-1.013}).
  Let us consider the case of $\kappa \rightarrow 0$ for $\kappa-\mu$. The PDF and Laplace transform of $G_i$ in (\ref{eq_secII-2.003}) and (\ref{eq_secII-2.010}) can be simplified for $\kappa \rightarrow 0$ as follows
    \begin{equation}
    \begin{split}
       f_{G_i}(x)
      &= \frac{1}{\Gamma(\mu)} \left( \frac{\mu}{\bar{w}} \right)^{\mu} x^{\mu-1}
      \exp\left( -\frac{\mu}{\bar{w}}x\right), \quad
    \mathcal{L}_{G_i}(s) = \frac{1}{\left( 1 + \frac{s \bar{w}}{\mu}\right)^{\mu}},
    \end{split}
    \label{eq_secV-1.007}
    \end{equation}
  where we used the following asymptotic property of the modified Bessel function in $f_{G_i}(x)$ \cite{Gradshteyn1994}
    \begin{equation}
    \begin{split}
      \lim_{z \rightarrow 0} \mathrm{I}_{\mu-1}(z) = \frac{1}{\Gamma(\mu)} \left(\frac{z}{2} \right)^{\mu-1}.
    \end{split}
    \label{eq_secV-1.006}
    \end{equation}
  Then, $c$ and $\varphi(r)$ in the interference of a D2D link (\ref{eq_secIII-2.002}) and a cellular link (\ref{eq_secIII-3.002})
  can be simplified as
    \begin{equation}
    \begin{split}
      c_{\kappa-\mu} &=
      \frac{q \varepsilon \lambda }{\xi} \cdot
      \frac{ \Hypergeometric{1}{1}{\mu+\delta_d}{\mu}{0}}{\mathrm{sinc}\left(\delta_d\right)} \cdot
      \left( \frac{\bar{w}}{\mu} \right)^{\delta_d} \cdot
      \binom{\mu+\delta_d-1}{\delta_d}\\
      &=
      \frac{q \varepsilon \lambda }{\xi \mathrm{sinc}\left(\delta_d\right)} \cdot
      \left( \frac{\bar{w}}{\mu} \right)^{\delta_d} \cdot
      \binom{\mu+\delta_d-1}{\delta_d},
        \end{split}
    \label{eq_secV-1.008}
    \end{equation}
      \begin{equation}
    \begin{split}
        \varphi(r) &=
        \int_{0}^{R} \mathcal{L}_{G}\left( s x^{\tau_c} {r}^{-\tau_c} \right) f_{L_c}(x) \mathrm{d}x
        =
        2 \pi \lambda_b \int_{0}^{R} \frac{x}{\left( 1 + \frac{s x^{\tau_c} {r}^{-\tau_c} \bar{w}}{\mu}\right)^{\mu}} \mathrm{d}x\\
        &=
        \Hypergeometric{2}{1}{\mu, \delta_c}{1 + \delta_c}{-\frac{s R^{\tau_c} {r}^{-\tau_c} \bar{w}}{\mu}},
      \end{split}
      \label{eq_secV-1.009}
      \end{equation}
      where $\delta_d = \frac{2}{\tau_d}$, $\delta_c = \frac{2}{\tau_c}$, and
      we used $\Hypergeometric{1}{1}{a}{b}{0} = 1$ in the second equality of (\ref{eq_secV-1.008}),
      applied (\ref{eq_secII-1.003}) and (\ref{eq_secV-1.007}) in the second equality of (\ref{eq_secV-1.009}),
      and used the following integration in the last expression \cite{Gradshteyn1994}
  \begin{equation}
    \begin{split}
       \int_{0}^{u} \frac{x^{\mu-1}}{\left( 1 + \beta x\right)^{\nu}} \mathrm{d}x =
       \frac{u^{\mu}}{\mu} \Hypergeometric{2}{1}{\nu, \mu}{1 + \mu}{-\beta u}.
      \end{split}
      \label{eq_secV-1.010}
      \end{equation}
  The average $\mathbb{E}\left[ g\left( \frac{W}{I+N_0}\right)\right]$ for $\kappa \rightarrow 0$ is derived as follows
   \begin{equation}
    \begin{split}
    \mathbb{E}\left[ g\left( \frac{W}{I+N_0}\right)\right] &=
    \mathbb{E}_{I}\left[
      \mathbb{E}\left[ \left. g\left( \frac{W}{I+N_0}\right) \right\vert I \right] \right]
      = \mathbb{E}_{I}\left[ \int_{0}^{\infty} g\left( \frac{x}{I+N_0}\right) f_{G_i}(x) \mathrm{d}x \right]\\
      &= g(0) + \frac{\bar{w}}{\mu N_0}
      \int_{0}^{\infty} g_{\mu}\left( \frac{\bar{w} x}{\mu N_0} \right)
      \mathcal{L}_{I}\left( \frac{x}{N_0} \right)
      \mathrm{e}^{-x} \mathrm{d}x\\
      &= g(0) + \frac{\bar{w}}{\mu N_0}
      \sum_{m = 1}^{M} c_m g_{\mu}\left( \frac{\bar{w} x_m}{\mu N_0} \right)
      \mathcal{L}_{I}\left( \frac{x_m}{N_0} \right)  + {R}_M,
    \end{split}
    \label{eq_secV-1.011}
    \end{equation}
  by using the similar procedures as Appendix VII, where $g_{\mu}(z) = \frac{1}{z}\left( 1 - \frac{1}{(1+z)^{\mu}}\right)$, and
  $c_m$, $x_m$, and $R_M$ are the $m$-th weight, abscissa, and remainder of the $M$-th order Laguerre polynomial, respectively.

  \bibliographystyle{IEEEtran}
  \bibliography{bib1}


  \end{document}